\documentclass{ieeeaccess}
\usepackage{cite}
\usepackage{amsmath,amssymb,amsfonts}
\usepackage{algorithmic}
\usepackage{graphicx}
\usepackage{textcomp}
\def\BibTeX{{\rm B\kern-.05em{\sc i\kern-.025em b}\kern-.08em
    T\kern-.1667em\lower.7ex\hbox{E}\kern-.125emX}}

\usepackage{dcolumn}
\usepackage{bm}

\usepackage{braket}
\usepackage{algorithm}
\usepackage{float}
\usepackage{stfloats}	
\usepackage{bbm}	

\DeclareMathOperator{\sign}{sign}

\DeclareMathOperator*{\argmin}{arg\,min}
\DeclareMathOperator{\GF}{GF}
\DeclareMathOperator{\tr}{tr}
\DeclareMathOperator{\Tr}{Tr}
\DeclareMathOperator{\row}{row}
\DeclareMathOperator{\wt}{wt}

\newcommand{\mbf}[1]{\pmb{#1}}

\newcommand{\teq}{\triangleq}

\newcommand{\sC}{{\cal C}}
\newcommand{\sG}{{\cal G}}
\newcommand{\sM}{{\cal M}}
\newcommand{\sN}{{\cal N}}
\newcommand{\sS}{{\cal S}}
\newcommand{\cS}{{\cal S}}

\newcommand{\CC}{{\mathbb C}}
\newcommand{\RR}{{\mathbb R}}
\newcommand{\ZZ}{{\mathbb Z}}

\let\emph\textit
\usepackage{amsthm}
\theoremstyle{definition}		

\newtheorem{definition}{Definition}

\newtheorem{theorem}{Theorem}
\newtheorem{corollary}{Corollary}
\newtheorem{example}{Example}
\newtheorem{remark}{Remark}

\def\approxprop{%
	\def\p{%
		\setbox0=\vbox{\hbox{$\propto$}}%
		\ht0=0.7ex \box0 }%
	\def\s{%
		\vbox{\hbox{$\sim$}}%
	}%
	\mathrel{\raisebox{0.4ex}{%
			\mbox{$\underset{\s}{\p}$}%
	}}%
}

\usepackage[normalem]{ulem}		

\renewcommand{\thevol}{000}
\renewcommand{\theyear}{2021}

\begin{document}

\history{Prepared Manuscript}
\doi{: TBD}

\title{Log-domain decoding of quantum LDPC codes over binary finite fields}

\author{	\uppercase{Ching-Yi~Lai} 
	~and~	\uppercase{Kao-Yueh~Kuo} 
}
\address{Institute of Communications Engineering, National Yang Ming Chiao Tung University, Hsinchu 30010, Taiwan.}

\tfootnote{CYL was  supported  by the Ministry of Science and Technology (MOST) in Taiwan, under Grant MOST110-2628-E-A49-007.}

\markboth
{C.-Y.~Lai and K.-Y.~Kuo: Log-domain decoding of quantum LDPC codes over binary finite fields}
{C.-Y.~Lai and K.-Y.~Kuo: Log-domain decoding of quantum LDPC codes over binary finite fields}

\corresp{Corresponding author: Ching-Yi~Lai (email: cylai@nycu.edu.tw).}

\begin{abstract}
A quantum stabilizer code over GF$(q)$ corresponds to a classical additive code over GF$(q^2)$ that is self-orthogonal with respect to a symplectic inner product. We study the decoding of quantum low-density parity-check (LDPC) codes over binary finite fields GF$(q=2^l)$ by the sum-product algorithm, also known as belief propagation (BP). Conventionally, a message in a nonbinary BP for quantum codes over GF$(2^l)$ represents a probability vector over GF$(2^{2l})$, inducing high decoding complexity. In this paper, we explore the property of the symplectic inner product and show that scalar messages suffice for BP decoding of nonbinary quantum codes, rather than vector messages necessary for the conventional BP. Consequently, we propose a BP decoding algorithm for quantum codes over GF$(2^l)$ by passing scalar messages so that it has low computation complexity. The algorithm is specified in log domain by using log-likelihood ratios (LLRs) of the channel statistics to have a low implementation cost. Moreover, techniques such as message normalization or offset can be naturally applied in this algorithm to mitigate the effects of short cycles to improve BP performance. This is important for nonbinary quantum codes since they may have more short cycles compared to binary quantum codes. Several computer simulations are provided to demonstrate these advantages. The scalar-based strategy can also be used to improve the BP decoding of classical linear codes over GF$(2^l)$ with many short cycles.
\end{abstract}


\titlepgskip=-15pt

\maketitle

\section{Introduction} \label{sec:Intro}
Quantum stabilizer codes can be used to protect quantum information with efficient encoding and decoding procedures similar to classical error-correcting codes \cite{Shor95,CS96,Steane96,GotPhD,CRSS98,NC00}. 
Since quantum coherence decays quickly, an efficient decoding procedure is particularly important.
For this purpose, sparse-graph quantum codes, similar to classical low-density parity-check (LDPC) codes \cite{Gal63,MN96,Mac99}, are preferred since they can be efficiently decoded by the sum-product algorithm  \mbox{\cite{MMM04, PC08,Wan+12,Bab+15,ROJ19,PK19,KL20, RWBC20 ,LP19}},
which is usually understood as a realization of belief propagation (BP) \cite{Pea88}. 
BP is appealing because  of   its good decoding performance for sparse-graph codes and efficient computation complexity that  is nearly linear in the code length~\cite{Gal63,MN96,Mac99}.
BP decoding is done by iteratively passing messages on a Tanner graph~\cite{Tan81} corresponding to the parity-check matrix of a code~\cite{WLK95,MMC98,AM00,KFL01,RU08}. 

The concept of stabilizer codes has been extended from the binary case (qubits) to nonbinary case ($q$-ary qudits) \cite{Kni96a,Kni96b,Rai99,MU00,AK01,KKKS06}.
We assume that a $q$-ary qudit  suffers errors from an error basis of $q^2$ elements.
A quantum stabilizer code over $\GF(q)$ can be considered as a classical additive code over $\GF(q^2)$.
Although binary quantum codes are most widely studied, a nonbinary quantum code has its advantages over binary ones. 
For example, binary stabilizer codes do not perform as well as nonbinary codes in spatially correlated noise \cite{KF05,NB06}.
To improve, we can group every $l$ qubits that are strongly-correlated together as a qudit and use quantum codes over $\GF(2^l)$
at a cost of higher decoding complexity. 

In BP decoding of a classical code over $\GF(q)$, a message is a probability vector of length $q$ for a $q$-ary variable~\cite{DM98}.
The complexity of BP is dominant by the \emph{check-node complexity} per edge for the convolution of messages,
which is $O(q\log q)$ per edge for classical codes over $\GF(q)$ \cite{MD01,DF07}.

For quantum codes over $\GF(q)$, the check-node complexity is even higher proportional to $q^2\log q^2$. 
	To decode a \mbox{length-$N$} quantum code over $\GF(q)$ with $M$ parity checks, one can represent its check matrix $H\in \GF(q^{2})^{M\times N}$ by a matrix in $\GF(q)^{M\times 2N}$ and decode the code using $q$-ary BP (BP$_q$), so that the check-node complexity reduces to $O(q\log q)$ per edge.
	If $q=2^l$, one can further represent $H$ by a binary matrix $\in\GF(2)^{M\times 2lN}$ to use BP$_2$ and the check-node complexity reduces to $O(1)$ per edge.
	However, these methods ignore some error correlations (such as the correlations between Pauli $X$ and $Z$ errors),
	causing performance degradation \cite{DT14,Bab+15,ROJ19,PK19,KL20}.
 Therefore, we would like to study the BP decoding problem for nonbinary quantum codes,
 and we will show that it can be simplified to $O(1)$ for binary fields $\GF(2^l)$ without ignoring the error correlations.
This is done by exploring the property of the symplectic inner product of quantum codes, so that  scalar messages suffice for BP decoding of nonbinary quantum codes, rather than  vector messages necessary for the conventional BP.\footnote{
	In Appendix~\ref{sec:CBP}, we clarify the complexities of our approach and the conventional approach, both maintaining the correlations between $q^2$ errors.}

	In classical coding theory, BP is usually implemented in log domain to require only additions and lookup tables for computation \cite{Gal63,WSM04,RVH95,HOP96,HEAD01} (see Remark~\ref{rmk:tbl} for our case). 
	Moreover, the required bit-width for each scalar variable is fewer in log domain compared to linear domain (Remark~\ref{rmk:cmpx2}). 
These greatly simplify the implementation of BP. 
We would like to have a quantum version of efficient log-domain BP.

 Previously we proposed a refined BP$_4$ for binary quantum codes \cite{KL20}, which uses scalar messages 
with check-node complexity $O(1)$ per edge and completes a decoding equivalent to the conventional BP$_4$.
In this paper, we extend this approach to quantum codes over $\GF(2^l)$ 
so that an efficient  $2^{2l}$-ary BP decoding with check-node complexity $O(1)$ per edge is possible.
The central idea is that a message of this BP concerns whether a qudit error commutes or anticommutes with a parity-check Pauli operator.  Consequently,  we are able to define a scalar message in log domain for quantum codes (Definition~\ref{def:la}) and propose a log-likelihood ratio BP \mbox{(LLR-BP)} decoding algorithm (Algorithm~\ref{alg:LLR-BP}). 
This algorithm is applicable to any quantum codes over $\GF(2^l)$.

The proposed algorithm could be extended to quantum codes over $\GF(p^l)$ for a prime $p$. 
In this case, the check-node complexity would reduce from $O(p^{2l}\log p^{2l})$ to $O(p\log p)$ per edge (see the discussion in~Sec.~\ref{sec:con}).

Another issue of BP decoding for quantum codes is that the Tanner graph of a stabilizer code contains many short cycles.\footnote{This is caused by the overlaps of the rows in the check matrix, which cannot be prevented due to the commutation relations of the stabilizers.} 
This may lead to ineffective message passing \mbox{\cite{MMM04,PC08,Wan+12,Bab+15,ROJ19,PK19,LP19,KL20,RWBC20}}
since short cycles introduce unwanted dependency between messages to affect the convergence of BP \cite[Sec~4.2]{Gal63}. 
The problem can be mitigated during the code construction \cite{MMM04,KHIS11,KP13,TZ14,PK19}  but this would restrict the code candidates in applications. 
Another direction is to improve the BP decoding algorithm.  
BP can be improved by post-processing \cite{PC08,Wan+12,Bab+15,ROJ19,PK19,RWBC20}, 
but this increases the computation complexity  and hence increases the decoding time. 
Another approach is to use a neural BP (NBP) decoder \cite{LP19}; however, NBP may not apply to large codes due to the complicated offline training process and more importantly, there is less guarantee of low error-floor from training (see Figs.~2 and~3).

As our LLR-BP uses scalar messages, it is straightforward to apply the techniques of message normalization or offset \cite{CF02b,CDE+05,YHB04} without incurring much additional cost. These techniques can suppress overestimated messages caused by short cycles \cite{YHB04}. Applying these techniques significantly improves the BP performance on binary quantum codes \cite{KL20}. 
For nonbinary quantum codes, the number of rows in a check matrix becomes more (e.g., see \eqref{eq:css_w} and \eqref{eq:css_g}), which may cause much more short cycles compared to the case of binary quantum codes.
However, applying message normalization or offset on the scalar-based nonbinary BP again significantly improves the performance,
and this improvement does not need a large number of iterations to achieve.
Computer simulations will be conducted to show these advantages.

Since quantum codes are like classical codes with short cycles, our approach can also be used to improve the BP decoding of classical codes over $\GF(2^l)$ with short cycles. (This will be discussed in Remark~\ref{rmk:add_vs_lin}.)

This paper is organized as follows. In Sec.~\ref{sec:stb}, we provide some basics for binary fields $\GF(2^l)$ and define stabilizer codes over $\GF(2^l)$.
In Sec.~\ref{sec:dec}, we  give a scalar-based \mbox{LLR-BP} for   stabilizer codes over $\GF(2^l)$. 
In Sec.~\ref{sec:sim}, simulation results for several stabilizer codes are provided. 
We conclude and give some discussions in Sec.~\ref{sec:con}.

\section{Stabilizer codes over binary finite fields} \label{sec:stb}

The theory of binary and $q$-ary quantum stabilizer codes can be understood as certain classical codes over finite fields $\GF(2^2)$ and $\GF(q^2)$, respectively \cite{Shor95,CS96,Steane96,GotPhD,CRSS98,NC00}, \cite{Kni96a,Kni96b,Rai99,MU00,AK01,KKKS06}. 
We refer to \cite{McE87} for the basics of finite fields.

\begin{definition} \label{def:Tr}
	For a finite field $\GF(p)$ and its extension field \mbox{$\GF(q'={p}^m)$}, the \emph{trace}   from $\GF(q')$ to $\GF(p)$ is defined as
	\begin{equation*}
	\Tr^{q'}_{p}(\eta) = \eta + \eta^p + \dots + \eta^{{p}^{m-1}} ~\in~ \GF(p)
	\end{equation*}
	for any $\eta\in\GF(q')$. 
\end{definition}

Assume $q=2^l$ in the following if not otherwise specified.

\subsection{Symplectic inner product  vector space} \label{sec:GF2_ext}

Let $\tr(\cdot)$ be the trace from $\GF(q)$ to the ground field $\GF(2) = \{0,1\}$, i.e., for  $\eta\in\GF(q)$,
	\begin{equation*}
	\tr(\eta) \triangleq \Tr^q_2(\eta) = \eta + \eta^2 + \dots + \eta^{2^{l-1}} ~\in~ \{0,1\}.
	\end{equation*}

For two vectors $u=(u_1,\dots,u_N),v=(v_1,\dots,v_N)\in\GF(q)^N$, their \emph{Euclidean inner product} is
$$u\cdot v = \textstyle \sum_{j} u_j v_j.$$

Let $\omega$ be a primitive element of $\GF(q^2)$. Then $\{\omega,\omega^q\}$   is a basis for $\GF(q^2)$ over $\GF(q)$ \cite{McE87,MU00,AK01,KKKS06}. 
Consequently,   a vector $u$ in the vector space $\GF(q^2)^N$ can be written as 
	\begin{align}
	u = \omega u' + \omega^q u'', 	\notag 
	\end{align}
	where $u', u'' \in\GF(q)^N$. 
	That is, $\GF(q^2)^N$ is isomorphic to $\GF(q)^{2N}$, and a $u\in\GF(q^2)^N$ is uniquely represented by a $(u'|u'')\in\GF(q)^{2N}$ given $\{\omega,\omega^q\}$.
	To better align with the notation in the next subsection, we denote $u'$ by $u^X$ and $u''$ by $u^Z$.
	Equivalently, $u\in\GF(q^2)^N$ may be denoted by 
	\begin{align}
		u = \omega u^X + \omega^q u^Z	~\equiv~  	(u^X|u^Z) 	 	~~\in~	\GF(q)^{2N}. \label{eq:vXZ}
	\end{align}
\begin{definition} \label{def:ip}
For $u,v\in\GF(q^2)^N$,  {where $q=2^l$,} the (binary) \emph{symplectic inner product} of $u$ and $v$ is defined as
	\begin{equation} \label{eq:ip}
	\langle u,v \rangle = \tr(u^X\cdot v^Z + u^Z\cdot v^X) ~\in~ \{0,1\}.
	\end{equation}

\end{definition}
Equation \eqref{eq:ip} is a symplectic $\{0,1\}$-bilinear form, or a symplectic inner product.
This form is important in the study of quantum codes, to link a quantum code over $\GF(q)$ as a classical code $C$ over $\GF(q^2)$ that is symplectic self-orthogonal (i.e., $\langle u,v \rangle= 0$ for any two codewords $u,v\in C$).

The \emph{Hermitian inner product} of $u,v\in\GF(q^2)^N$ is
	$$u\cdot \bar v = \textstyle \sum_j u_j v_j^q,$$
where $\bar v \triangleq (v_1^q,v_2^q,\dots,v_N^q)$.

\begin{theorem} \label{thm:Herm} \cite{MU00,KKKS06}
	For $u,v\in\GF(q^2)^N$,  {where $q=2^l$,} their symplectic inner product \eqref{eq:ip} can be computed by 
	\begin{equation}\label{eq:tr_hm}
	\langle u,v \rangle = \tr\left(\frac{u\cdot \bar v + v\cdot \bar u}{\omega^2 + \omega^{2q}}\right).
	\end{equation}
\end{theorem}
\begin{proof}
	Write $u=\omega u^X+\omega^q u^Z$ and $v=\omega v^X+\omega^q v^Z$. Then
	\begin{equation} \label{eq:hm2}
	\begin{aligned}
	u\cdot \bar v + v\cdot \bar u &= u\cdot \bar v  + \overline{u\cdot \bar v}  \\ 
	&=(\omega^2 + \omega^{2q})(u^X\cdot v^Z + u^Z\cdot v^X),
	\end{aligned}
	\end{equation}
	which directly leads to \eqref{eq:tr_hm}.
\end{proof}
A consequence of Definition~\ref{def:ip} and Theorem~\ref{thm:Herm} is:
\begin{corollary} \label{col:h2ip}
	 {For $q=2^l$,} two vectors ${u,v\in\GF(q^2)^N}$ satisfy ${\langle u,v \rangle=0 }$ if one of the following conditions holds:
	 \begin{enumerate}
	 	\item   They are Hermitian orthogonal, i.e.,
	 	\begin{equation} \label{eq:h_0}
	 	u\cdot \bar v = 0.
	 	\end{equation}
		\item  They satisfy the \emph{CSS conditions} \mbox{\cite{CS96,Steane96}}:
		\begin{equation} \label{eq:css_0}
		\begin{aligned}
		& u\equiv(u^X|\,\mbf 0),\, v\equiv(v^X|\,\mbf 0); \\ 
		&\text{or~~} u\equiv(\mbf 0\,|u^Z),\, v\equiv(\mbf 0\,|v^Z); \\ 
		&\text{or~~} u\equiv(u^X|\,\mbf 0),\, v\equiv(\mbf 0\,|v^Z) \text{~with~} u^X\cdot v^Z =0. 
		\end{aligned}
		\end{equation}
	 \end{enumerate}
\end{corollary}

\begin{example} \label{ex:q=2}
	If $q=2$, we have $\omega^2 + \omega^{2q} = 1$ in $\GF(4)$ and $\tr(\cdot)=\Tr_2^2(\cdot)$ is the identity map. 
		For $u,v\in\GF(4)^N$,
		\begin{align*}
		\langle u,v \rangle = u^X\cdot v^Z + u^Z\cdot v^X.
		\end{align*}
	It is well-known that $\langle u,v \rangle = \Tr^4_2(u\cdot \bar v)$ \cite{CRSS98}. 
	Here, $\Tr^4_2(0)=\Tr^4_2(1)=0$ and $\Tr^4_2(\omega)=\Tr^4_2(\omega^2)=1$ by Definition~\ref{def:Tr}.
\end{example}

	Given   $\eta \in \GF(q^2)$, define
	\begin{align*}
	[\eta] \,\, &= \{\xi\in\GF(q^2) : \langle \eta ,\xi\rangle = 0 \}, \\
	[\eta]^{\text{c}} &= \{\xi\in\GF(q^2) : \langle \eta ,\xi\rangle = 1 \}.
	\end{align*}
Then $ [\eta]$ and  $[\eta]^{\text{c}}$ are two sets with elements that commute and anticommute with a certain element $\eta$, respectively, and $\{[\eta], [\eta]^{\text{c}} \}$ is a partition of  $\GF(q^2)$. 
\begin{example} \label{ex:q^2=4}
	For the special case of $q=2$, $\GF(4)=\{0,~ 1,~\omega,~\omega^2\}$ with $1=\omega+\omega^2$. 
	Then $\langle 1,\omega \rangle= \langle \omega,\omega^2 \rangle=\langle 1,\omega^2 \rangle=1$.
	Thus $[\omega]= \{0,\omega\}$ and $[{\omega}]^{\text{c}}= \{1,\omega^2\}$.
\end{example}

	For two sets $A$ and $B$, define $A\setminus B = \{a\in A: a\notin B\}$. 
	If $B=\{b\}$, we may write $A\setminus\{b\}$ as $A\setminus b$ for simplicity.

\begin{theorem} \label{thm:tr_0.5}
	For $q=2^l$,  there exists an additive subgroup $S\subset \GF(q)$ of size $q/2$ such that  for {$\eta\in S$}, $\tr(\eta)=0$ and
	for $\eta' \in \GF(q)\setminus S$, $\tr(\eta')=1$.
\end{theorem}
\begin{proof}
	All the elements in \mbox{$\GF(q)$} are the roots of $x^{2^l}+x$, which has a factorization 
	$$x^{2^l}+x = \textstyle \prod_{b\in\{0,1\}}(x + x^2 + \dots + x^{2^{l-1}} + b),$$
	where the each factor $(x + x^2 + \dots + x^{2^{l-1}} 
	+ b)$ contains $2^{l-1} = q/2$ distinct roots in $\GF(q)$ for $b=0$ or $1$.
	If $\eta$ is a root of  $(x + x^2 + \dots + x^{2^{l-1}} 
	+ b)$, then it can be shown that  $\tr(\eta)=b$. 
Thus we have the statement.
\end{proof}
{Note that Theorem~\ref{thm:tr_0.5} easily extends for $\GF(p^l)$ by a similar factorization equality in \cite[Theorem 8.1(e)]{McE87}.}

\begin{corollary} \label{col:tr_half}
	Given  $\eta\ne 0$ in $\GF(q^2)$, {where $q=2^l$,} 
	$$|[\eta]|=|[\eta]^{\text{c}}|=q^2/2.$$
\end{corollary}
\begin{proof}
	Let $\eta= (\eta^X|\eta^Z)\in\GF(q^2)$ for $\eta^X,\eta^Z\in\GF(q)$.
	Since $\eta\neq 0$, at least $\eta^X\neq 0$ or $\eta^Z\neq 0$.	
	Suppose that $\eta^Z\neq 0$ without loss of generality. 
	For $\xi= (\xi^X|\xi^Z)\in\GF(q^2)$ with $\xi^X,\xi^Z\in\GF(q)$,
	we have $\langle \eta,\xi \rangle =  \tr(\eta^X \xi^Z+\eta^Z \xi^X)= \tr(\eta^X \xi^Z)+\tr(\eta^Z \xi^X)$.
	Then $\{ \eta^Z \xi^X: \xi^X\in\GF(q)\}$ is a permutation of $\{\xi^X\in\GF(q)\}$. 
	By Theorem~\ref{thm:tr_0.5}, there exists  $S\subset \{ \eta^Z \xi^X: \xi^X\in\GF(q)\}$  with $|S|=q/2$
	such that for all $\mu\in S$, $\tr(\mu)=0$ and for $\mu'\in\GF(q)\setminus S$, $\tr(\mu')=1$.
	If $\eta^X=0$, then we are done.
	Consider $\eta^X\neq 0$. 
	Note that a fixed $\nu\in\{\eta^X \xi^Z:\xi^Z\in\GF(q)\}$ is paired with each
	${\chi}\in \{\eta^Z \xi^X: \xi^X\in\GF(q)\}$ such that
		\[\tr(\nu)+\tr(\chi)=
		\begin{cases}
		\tr(\nu),& \hbox{ if $\chi\in S$};\\
		\tr(\nu)+1,& \hbox{ if $\chi\in\GF(q)\setminus S$.} \\
		\end{cases}
		\]	
	For a fixed $\nu$, a half of the pairs ${ \{(\nu,\chi)\} }$ have $ \tr(\nu)+\tr(\chi)=0 $, 
	while the other half have $ \tr(\nu)+\tr(\chi)=1 $.
	This holds for any $\nu$. Thus $|[\eta]|=|[\eta]^{\text{c}}|=q^2/2.$
\end{proof}

The property that $|[\eta]|=|[\eta]^c|$ has its merit in implementation (for an important operator in Definition~\ref{def:la}).

\newcommand*\GFtwol{ $\GF(2^l)$ }						
\subsection{Stabilizer codes over \protect\GFtwol}		

We review some basics of nonbinary quantum codes \cite{Kni96a,Kni96b,Rai99,MU00,AK01,KKKS06} and then define the decoding problem.

Let $\RR$ be the field of real numbers, $\ZZ_+$ be the set of positive integers, and $\CC$ be the field of complex numbers.

Consider a $q$-ary quantum system with $q=2^l$, for some integer $l\ge 1$.
A \emph{qudit} is a unit vector in~$\CC^q$. 
Without loss of generality, let \mbox{$\{\ket \eta : \eta\in\GF(q)\}$} be a set of orthonormal basis for $\CC^q$ such that
there is a set of (generalized) Pauli operators $\{ X(\xi), Z(\xi'): \xi,\xi'\in\GF(q) \}$ on~$\CC^q$,   where
\begin{equation} \label{eq:basis}
X(\xi)\ket \eta = \ket{\eta+\xi}, \quad Z(\xi)\ket \eta = (-1)^{\tr(\eta \xi)}\ket \eta.
\end{equation}
Note that $X(0) = Z(0) = I$, the identity.

An $N$-qudit state is a unit vector in $(\CC^q)^{\otimes N} = \CC^{q^N}$.
For $u= (u_1,\dots,u_N)\in\GF(q)^N$, we define the (generalized) $N$-fold Pauli operators:
	$$X(u) = X(u_1)\otimes\cdots\otimes X(u_N), ~~ Z(u) = Z(u_1)\otimes\cdots\otimes Z(u_N).$$
	%
The set  $\{X(u^X)Z(u^Z): u^X,u^Z\in\GF(q)^N\}$ forms a  basis for the linear operators on $\CC^{q^N}$.
It suffices to consider these discrete errors for error correction according to the error discretization theorem \cite{NC00}.
To form a group of operators with closure, we include phases $\pm 1$ (due to $(-1)^{\tr(\eta \xi)}\in\{\pm 1\}$ in~\eqref{eq:basis} for $q=2^l$) and consider the group
	$$\sG_N \teq \{\pm X(u^X)Z(u^Z): u^X,u^Z\in\GF(q)^N\}.$$
Elements of $\sG_N$ are connected to the elements of the vector space $\GF(q^2)^N$ 
by a homomorphism $\varphi:\sG_N \to \GF(q^2)^N$ defined as follows.  
For   $E = (-1)^c \otimes_{n=1}^N \left(X(u^X_{n}) Z(u^Z_{n})\right)\in\sG_N$, where $c\in\{0,1\}$ and $u^X,u^Z\in\GF(q)^N$,
\begin{equation} \label{eq:E2u}
\varphi(E) \teq \omega u^X + \omega^{q} u^Z \equiv (u^X|u^Z),
\end{equation}
where $\omega$ is a primitive element of $\GF(q^2)$ as discussed in the previous subsection.
Note that the kernel of $\varphi$ is $\{\pm I^{\otimes N}\}$, i.e., the phase $(-1)^c$ of $E$ does not appear in $\varphi(E)$.

Two operators $E,F\in\sG_N$ either \emph{commute} \mbox{($EF= FE$)} or \emph{anticommute} \mbox{($EF = - FE$)} with each other.
Suppose that $\varphi(E)=(u^X|u^Z)$ and $\varphi(F)=(v^X|v^Z)$. Then
\begin{align*}
EF &= (-1)^{\tr(u^X\cdot v^Z + u^Z\cdot v^X)}FE\\
   &= (-1)^{\langle \varphi(E),\varphi(F) \rangle}FE.
\end{align*}

\begin{definition} \cite{AK01,KKKS06}
	A \emph{stabilizer group} $\sS$ is an Abelian subgroup of $\sG_N$ such that  $-I^{\otimes N} \notin \sS$.
	The \emph{stabilizer code} $\sC(\sS)$ defined by $\sS$ is a subspace of $\CC^{q^N}$ that is the joint \mbox{$(+1)$-eigenspace} of all the operators in~$\sS$, i.e.,
	\begin{equation*}
	\sC(\sS) = \{\ket{\psi} \in \CC^{q^N}:\, F\ket{\psi} = \ket{\psi} ~\forall\, F\in\sS \}.
	\end{equation*}
	Every element in $\sS$ is called a \emph{stabilizer}.
\end{definition}

If an occurred error $E\in\sG_N$ anticommutes with certain stabilizers, it can be detected by measuring the eigenvaules of an independent generating set of $\cS$.
If any of the eigenvalues is $-1$, we know that there is an error. 
Let $ \sS^{\perp} = \{E'\in\sG_N: E'F = FE' ~~\forall~ F\in\sS\} $.
Apparently,  if $E\in \cS^\perp\setminus \pm \sS$, then it cannot be detected and will lead to a logical error.

For $u\in\GF(q^2)^N$, let $|u|$ be \emph{Hamming weight} of $u$. 
For $E\in\sG_N$, define the \emph{weight} of $E$ as {$\wt(E) \teq |\varphi(E)|$}.
The \emph{minimum distance} of $\sC(\cS)$ is defined as $$D\teq \min\{\wt(F):\, F \in  { \cS^\perp \setminus \pm\cS }  \}.$$ 
	$\sC(\sS)$ can correct an arbitrary error $E\in\sG_N$ with weight $\wt(E) \le \lfloor\frac{D-1}{2}\rfloor$.
	If $\sS$ has $l(N-K)$ independent generators, then $\sC(\sS)$ has dimension $q^{K}$.
In this case, $\sC(\sS)$ is called an $[[N,K,D]]_q$ stabilizer code.

Equivalently we may study quantum stabilizer codes using the language of finite fields. 
Let $C \triangleq\varphi(\sS) \,\subset\, \GF(q^2)^N.$ 
Since $\sS$ is an Abelian group, $C$ is an additive code over $\GF(q^2)$ that is self-orthogonal with respect to $\langle\cdot,\cdot\rangle$ in \eqref{eq:ip}, i.e,
$C$ is contained in the symplectic dual $C^\perp$ of $C$,
\begin{equation*}
C \subseteq C^\perp \teq \{u\in \GF(q^2)^N :\, \langle u, v \rangle = 0\, ~\forall~ v\in C \}.
\end{equation*}
Thus $\varphi(\sS^\perp) = C^\perp$ and 
	$
	D = \min\{ |v| :\, v\in C^\perp \setminus C\}.
	$

The encoding and decoding procedures of $\sC(\sS)$ are strongly related to $C$ and $C^\perp$. 
The encoding procedure,  {as well as the measurement circuit,} can be referred to \cite{NC00,GRB03}.
We discuss the decoding procedure in the following, which is similar to the classical syndrome-based decoding of $C^\perp$.

The phase of a Pauli operator can be ignored when discussing the decoding, so, without loss of generality, in the following we consider errors or stabilizers with phase $+1$.

Suppose that a codeword (state) in $\sC(\sS)$ is disturbed by an \mbox{unknown} error $E = E_1\otimes\cdots \otimes E_N\in\sG_N$.
We would like to infer $e = \varphi(E) \in\GF(q^2)^N$ by measuring 
a sequence of $M$ stabilizers  $\{S_m\}_{m=1}^M$, where $M\ge l(N-K)$ since $\cS$ has $l(N-K)$ independent generators. 
Write 
\begin{equation} \label{eq:S_m}
S_m = S_{m1} \otimes \cdots \otimes S_{mN},
\end{equation}
where $S_{mn}\in\sG_1$.
Since the error either commutes or anticommutes with a stabilizer, measuring a stabilizer returns outcome $+1$ or $-1$, which gives us a bit of error information. 
Let $H\in\GF(q^2)^{M\times N}$ with $(m,n)$-entry $H_{mn} = \varphi(S_{mn})$. 
Then $H$ is called a \emph{check matrix} of the stabilizer code $\sC(\sS)$, and $C$ is the row space of $H$, denoted by $C=\row(H)$.
After the measurements defined by $\{S_m\}_{m=1}^M$, we will obtain a \emph{binary error syndrome} $z = (z_1,\dots,z_M) \in\{0,1\}^M$, where
\begin{equation} \label{eq:zm}
z_m = \langle \varphi(E), \varphi(S_m) \rangle = \langle e, H_m \rangle,  
\end{equation}
where $H_m\in\GF(q^2)^N$ is the $m$-th row of $H$, called a \emph{check}. 
The decoding problem is as follows. 

\vspace*{4pt}
{\bf Decoding a stabilizer code over $\GF(q=2^l)$}:
Given a check matrix $H\in\GF(q^2)^{M\times N}$, a binary syndrome $z\in\{0,1\}^M$ of some (unknown) $e\in\GF(q^2)^N$, and certain characteristics of the error model, 
the decoder has to infer an $\hat e\in\GF(q^2)^N$ such that $\langle \hat e, H_m \rangle = z_m$ for $m=1,2,\dots,M$ and $\hat e - e \in \row(H)$ 
 with  probability as high as possible.
\vspace*{4pt}

\Figure[t!][width=8.0cm]{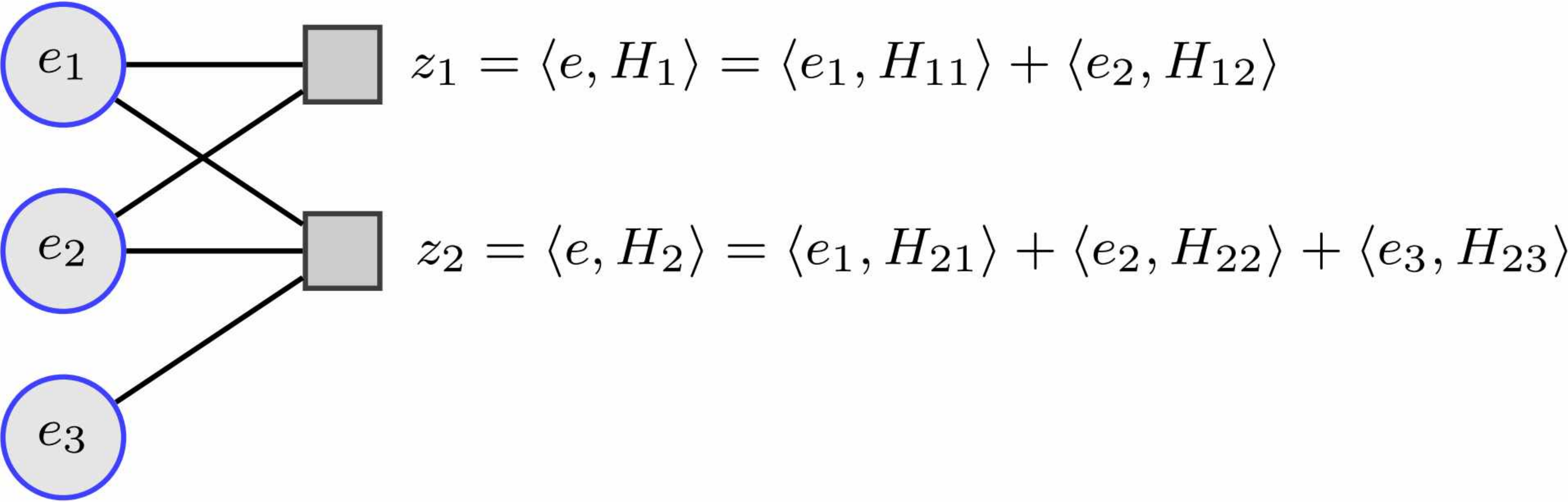}
{The Tanner graph of $H = \left[ {H_{11}\atop H_{21}} {H_{12}\atop H_{22}} {0\atop H_{23}} \right]$.\label{fig:H_2x3}}

Decoding by BP needs a Tanner graph defined by the check matrix $H$ of the code $\sC(\sS)$.
The Tanner graph  is a bipartite graph consisting of $N$ \emph{variable nodes} (representing $\{e_n\}_{n=1}^N$) and $M$ \emph{check nodes} (representing $\{z_m\}_{m=1}^M$), and there is an edge, with \emph{type} $H_{mn}$, connecting variable node~$n$ and check node~$m$ if \mbox{$H_{mn}\ne 0$}. 
An example is shown in Fig.~\ref{fig:H_2x3}.
Let 
\begin{equation*}  
\begin{aligned}
& \sN(m) = \{n : H_{mn} \ne 0\},\\
& \sM(n) = \{m : H_{mn} \ne 0\},
\end{aligned}
\end{equation*}
which are the neighboring sets of the check node $m$ and the variable node $n$, respectively.  
BP infers $\{e_n\}_{n=1}^N$ by passing messages between neighboring nodes in the Tanner graph.

\subsection{Stabilizer code constructions}

We finish this section with some stabilizer code constructions that will be used in our discussions and simulations.

	A check matrix $H\in\GF(2^{2l})^{M\times N}$ of a $2^l$-ary quantum code may be constructed from an $\frac{M}{2l}\times N$  parity-check matrix $\tilde H$ of a classical linear code over $\GF(2^{2l})$ 
such that $\tilde H$ is Hermitian self-orthogonal (cf.~Corollary~\ref{col:h2ip}).
Given such an $\tilde H$, we have a check matrix $H$ by the \emph{CSS extension}:
\begin{equation} \label{eq:css_w}
H = \left[ \begin{smallmatrix} \tilde H \\ \omega\tilde H \\ \vdots\\ \omega^{2l-1}\tilde H \end{smallmatrix} \right],
\end{equation}
where~$\omega$~is a primitive element of $\GF(2^{2l})$.
A simple construction of \eqref{eq:css_w} is to find a binary $\tilde H$ such that ${\tilde H \tilde{H}^T = O}$,
where $O$ is a zero matrix of appropriate dimension.\footnote{
	If $\tilde H$ is binary and $\tilde H \tilde{H}^T = O$, the supports of any two rows $u,\, v$ of $\tilde H$ have an overlap of even size; 
	consequently, any two rows $\omega^i u,\, \omega^j v$ of $H$ in \eqref{eq:css_w} satisfy $(\omega^i u)\cdot\overline{(\omega^j v)} = 0$ 
	(i.e. they satisfy~\eqref{eq:h_0}, as required).
	}

If $\tilde H$ in \eqref{eq:css_w} is replaced by a set of Euclidean orthogonal binary matrices 
	$\{H^{(i)}\}_{i=0}^{2l-1}$, 
	where \mbox{$H^{(i)}(H^{(j)})^T = O$} for $i\ne j$, 
	then we have the \emph{generalized CSS extension}:\footnote{
		We do not need $H^{(i)} (H^{(i)})^T = O$, since two rows $u,v$ of $H^{(i)}$ become two rows $\omega^i u, \omega^i v$ of $\omega^i H^{(i)}$ 
		in \eqref{eq:css_g}, with $\langle \omega^i u, \omega^i v \rangle =0$ by Theorem~\ref{thm:Herm}.
		}
		%
		%
		%
	\begin{equation} \label{eq:css_g}
	H = \left[ \begin{smallmatrix} H^{(0)} \\ \omega H^{(1)}  \\ \vdots\\ \omega^{2l-1} H^{(2l-1)}  \end{smallmatrix} \right].
	\end{equation}

Note that although \eqref{eq:css_w} and \eqref{eq:css_g} are called (generalized) CSS extension, the rows of $H$ may not necessarily have the form defined in the CSS conditions~\eqref{eq:css_0}.   

\section{BP decoding of stabilizer codes} \label{sec:dec}

	In the following discussion we assume that  
		$$\GF(q^2)=\{0,1,\zeta, \zeta^2,\dots, \zeta^{q^2-2}\},$$  
	where $\zeta$ is a primitive root of $\GF(q^2)$ and $q=2^l$. 

%
Consider that we want to solve the quantum decoding problem mentioned after \eqref{eq:zm}.
First, express \eqref{eq:zm} as
\begin{equation} \label{eq:z_m}
z_m = \langle e,H_m \rangle = \textstyle \sum_{n=1}^N \langle e_n, H_{mn} \rangle  \quad \text{(mod 2)}.
\end{equation}
The addition  for syndrome generation will always be mod~2 without further specification.
We may use $z = \langle e,H \rangle$ to mean that $z_m = \langle e,H_m \rangle ~\forall~ m$.

Consider a memoryless error model that each qudit suffers an independent error, 
i.e., the probability that the error vector is $u\in\GF(q^2)^N$ is $P(e=u) = \prod_{n=1}^N P(e_n=u_n)$
for some distribution ${P(e_n=\eta)=p_\eta}$ over $\{\eta\in \GF(q^2)\}$,  
where ${\sum_\eta p_\eta = 1}$.
In BP, the goal is to find the most probable value of each $e_n\in\GF(q^2)$ given $z$. 
Initially each $e_n$ has value in $\GF(q^2)$ according to the initial channel characteristics ${ \{P(e_n=\eta) = p_\eta\}_{\eta\in\GF(q^2)} }$. 
The likelihood of each value of $e_n$ is recorded and continuously updated according to (\ref{eq:z_m}) and the initial channel characteristics.

\subsection{LLR-BP for stabilizer codes over \GFtwol} \label{sec:alg}	

\begin{definition} \label{def:LLR}
Suppose that we have an (unknown) error vector $e=(e_1,\dots,e_N)\in \GF(q^2)^N$.
	For  $n=1,\dots,N$, the initial belief  of  $e_n$ being 0 rather than  $\zeta^i\neq 0\in\GF(q^2)$ for $i\in\{0,1,\dots,q^2- {2}\}$ is the  \emph{log-likelihood ratio} (LLR)
	\begin{equation*}
	\Lambda_n^{(i)} = \ln\frac{P(e_n=0)}{P(e_n=\zeta^i)} \in \RR.
	\end{equation*}
	The initial belief distribution of $e_n$ is the LLR vector 
	\begin{equation*}
	\Lambda_n = (\Lambda_n^{( {0})}, \Lambda_n^{( {1})}, \dots, \Lambda_n^{(q^2- {2})}) \in\RR^{q^2-1}.
	\end{equation*}
\end{definition}

\begin{definition} \label{def:ch}
	A \emph{generalized depolarizing channel} with error rate $\epsilon$ is a memoryless $q^2$-ary symmetric channel that 
	takes each nonzero error $\zeta^i \in \GF(q^2)$, $i\in\{0,1,\dots,q^2-2\}$, with probability $\frac{\epsilon}{q^2-1}$ and no error with probability $1-\epsilon$, i.e.,
	\begin{equation*}
	\Lambda_n^{(i)} = \ln\left( (1-\epsilon)(q^2-1)/\epsilon \right), ~~\forall~~ i=0,1,\dots,q^2-2.
	\end{equation*}
\end{definition}

 The initial beliefs $\{\Lambda_n\}_{n=1}^N$ are the information about the channel statistics and will be kept constant during decoding.

\begin{figure*}
\normalsize			
\newcounter{MYtempeqncnt}
\setcounter{MYtempeqncnt}{13}
\setcounter{equation}{17}
	~\\[-30pt]	
	\begin{align}
	\Gamma_n^{(i)} &=\Lambda_n^{(i)} + \sum_{m\in\sM(n)}\ln \frac{
	\sum_{u|_{\sN(m)\setminus n}: \atop { \langle u|_{\sN(m)\setminus n}, H_m|_{\sN(m)\setminus n} \rangle \atop = \langle 0, H_{mn} \rangle +z_m } } \left(\prod_{n'\in\sN(m)\setminus n} P(e_{n'}=u_{n'})\right) }{ 
	\sum_{u|_{\sN(m)\setminus n}: \atop { \langle u|_{\sN(m)\setminus n}, H_m|_{\sN(m)\setminus n} \rangle \atop = \langle \zeta^i, H_{mn} \rangle +z_m } } \left(\prod_{n'\in\sN(m)\setminus n} P(e_{n'}=u_{n'})\right) }  \label{eq:v_last0}\\
	&= \Lambda_n^{(i)} + \sum_{m\in\sM(n) \atop  \langle \zeta^i, H_{mn} \rangle=1  } (-1)^{z_m} \ln \frac{
	\sum_{u|_{\sN(m)\setminus n}: \atop { \langle u|_{\sN(m)\setminus n}, H_m|_{\sN(m)\setminus n} \rangle = 0 } } \left(\prod_{n'\in\sN(m)\setminus n} P(e_{n'}=u_{n'})\right) }{
	\sum_{u|_{\sN(m)\setminus n}: \atop { \langle u|_{\sN(m)\setminus n}, H_m|_{\sN(m)\setminus n} \rangle = 1 } } \left(\prod_{n'\in\sN(m)\setminus n} P(e_{n'}=u_{n'})\right) }. \label{eq:v_last}
	\end{align}
\setcounter{equation}{\value{MYtempeqncnt}}	
\hrulefill									
\end{figure*}

Given ${\{\Lambda_n\}_{n=1}^N,\ {z\in\{0,1\}^M}, \text{ and } H\in\GF(q^2)^{M\times N} }$,
BP intends to find the most probable value of  $e_n$ for each $n$ by computing 
a set of \emph{running beliefs} $ \{\Gamma_n = (\Gamma_n^{(0)}, \dots, \Gamma_n^{(q^2-2)})\in\RR^{q^2-1}\}_{n=1}^N $, 
where $\Gamma_n^{(i)}$ is to estimate
\begin{align}
 &\ln\frac{ P(e_n=0\,\,\mid  \text{syndrome } z) }{ P(e_n=\zeta^i\mid  \text{syndrome } z  )}\notag 
 =\ln \frac{ P(e_n=0,\,\,\, \text{syndrome }z) }{ P(e_n=\zeta^i,\, \text{syndrome } z )} \notag \\ 
 &=\ln \frac{ \sum_{u\in\GF(q^2)^N: u_n = 0}      ~\, \mathbbm 1_{(\langle u,H \rangle=z)} \, P(e=u) }
			{ \sum_{u\in\GF(q^2)^N: u_n = \zeta^i} \, \mathbbm 1_{(\langle u,H \rangle=z)} \, P(e=u) }, \label{eq:bp1} 
\end{align}
where $\mathbbm 1_{(\langle u,H \rangle=z)} = 1$, if $\langle u,H \rangle = z$, and \mbox{$\mathbbm 1_{(\langle u,H \rangle=z)} = 0$}, otherwise.
For $u\in\GF(q^2)^N$, let $u|_{\sN(m)}$ be the restriction of $u$ to $\sN(m)$. 
Note that $H_m|_{\sN(m)}$ is the vector consisting of the nonzero entries of $H_m$.
Then \eqref{eq:z_m} can be written as
\begin{align} 
z_m &= \langle ~ e|_{\sN(m)},~ H_m|_{\sN(m)} ~ \rangle, ~~ \text{or} \label{eq:chk} \\
\langle e_n, H_{mn} \rangle + z_m &= \langle ~ e|_{\sN(m)\setminus n},~ H_m|_{\sN(m)\setminus n} ~ \rangle \label{eq:upd_n}.
\end{align}
By~\eqref{eq:chk} and \eqref{eq:upd_n}, terms in \eqref{eq:bp1} can be approximated by the distributive law \cite{AM00}, 
if the check matrix is sparse, as follows: 
\small
\begin{align}
& P(e_n=\zeta^i,\, \text{syndrome } z) \notag\\ 
=& \sum_{u\in\GF(q^2)^N: u_n=\zeta^i, \atop \langle u, H_m \rangle = z_m \forall m} P(e=u) \notag \\
=& \sum_{u\in\GF(q^2)^N: u_n=\zeta^i, \atop \langle ~ u|_{\sN(m)},~ H_m|_{\sN(m)} ~ \rangle = z_m \forall m} \bigg( \textstyle \prod\limits_{j=1}^N P(e_{j}=u_{j}) \bigg) \notag \\
\approxprop & { \textstyle \prod\limits_{m\in\sM(n)} \Bigg( \sum_{u|_{\sN(m)\setminus n}: \atop { \langle u|_{\sN(m)\setminus n}, H_m|_{\sN(m)\setminus n} \rangle \atop =\, \langle \zeta^i, H_{mn} \rangle \,+\, z_m } } P(e|_{\sN(m)\setminus n}=u|_{\sN(m)\setminus n}) \Bigg) } \notag\\\
&\times P(e_n=\zeta^i), \label{eq:sum_all}
\end{align}
\normalsize
where ${ P(e|_{\sN(m)\setminus n}=u|_{\sN(m)\setminus n}) = \textstyle \prod\limits_{n'\in\sN(m)\setminus n} P(e_{n'}=u_{n'}) }$ 
after $P(e_n=\zeta^i)$ is separated.\footnote{
	 To get the result of \eqref{eq:sum_all}, one can also start from the Bayes' theorem: 
		${ P(e_n=\zeta^i,\, \text{syndrome } z) = P(\text{syndrome } z \mid e_n=\zeta^i) \times P(e_n=\zeta^i) }$.
	}
It is similar for approximating ${ P(e_n=0,\, \text{syndrome } z) }$.
When the Tanner graph is a tree, 
the (proportional) approximation in \eqref{eq:sum_all} becomes an equality after iterative updates (by a flow described after~\eqref{eq:v_by_tanh}).
The approximation is usually good for a sparse $H$, and then we have
the $\Gamma_{n}^{(i)}$ for estimating \eqref{eq:bp1} expressed as in \eqref{eq:v_last0} and \eqref{eq:v_last} at the top of this page.

\setcounter{equation}{19}				

The numerator (or denominator) in the logarithm of~\eqref{eq:v_last} is close to the usual sum-product computation of BP over $\GF(2)$, which has efficient hyperbolic tangent rule for computation \cite{Gal63}, \cite[Sec.~2.5.2]{RU08}; 
however, $u|_{\sN(m)\setminus n}$ here is a vector in $\GF(q^2)^{|\sN(m)|-1}$, rather than $\GF(2)^{|\sN(m)|-1}$.
Thus it needs simplification before we can reach an efficient computation.
The trick is to describe the likelihood of commutation and anti-commutation of $e_n$ and $H_{mn}$ by a Bernoulli (binary) random variable.
At qudit $n$, if $H_{mn}=\eta\ne 0$, then the required random variable would be $\langle e_n, \eta \rangle \in\{0,1\}$.

\begin{definition} \label{def:la}
For an LLR-type vector, such as $\Lambda_n$ in Def.~\ref{def:LLR}, 
we define a belief-quantization operator $\lambda_\eta:\RR^{q^2-1}\rightarrow \RR$  for $\eta\in { \GF(q^2)\setminus\{0\} }$ by
 \begin{equation} \label{eq:la}
 \begin{aligned}
 \lambda_{\eta}(\Lambda_n) 
 &= \ln\frac{ \sum_{\xi\in\GF(q^2): \langle \xi,\eta \rangle = 0} P(e_n=\xi) }{ \sum_{\xi\in\GF(q^2): \langle \xi,\eta \rangle = 1} P(e_n=\xi) }\\
 &= \ln\frac{ 1 + \sum_{i: \langle \zeta^i,\eta \rangle = 0} e^{-\Lambda_n^{(i)}} }{ \sum_{i: \langle \zeta^i,\eta \rangle = 1} e^{-\Lambda_n^{(i)}} }.
 \end{aligned}
 \end{equation}
\end{definition}
Note that 
\begin{equation} \label{eq:la_1}
\lambda_{\eta}(\Lambda_n) = \ln\frac{ P(\langle e_n, \eta \rangle = 0) }{ P(\langle e_n,\eta \rangle = 1) }, 
\end{equation}
which is the LLR of the binary random variable  $\langle e_n, \eta \rangle$
and this term features the major difference  between our algorithm and the classical nonbinary LLR-BP \cite{WSM04}.

Having the scalar information $\lambda_{H_{mn}}(\Lambda_n)$ for each edge $(m,n)$, where ${ H_{mn}\ne 0 }$, BP completes the update \eqref{eq:v_last} by the $\tanh$ rule mentioned in \cite{HOP96,KFL01}, \cite[Sec.~2.5.2]{RU08}, defined as follows.
	For two scalars $x,y\in\RR$,
	\begin{align}
	x \boxplus y  &\teq 2\tanh^{-1}\left( \tanh\frac{x}{2} \times \tanh\frac{y}{2} \right). \label{eq:bplus}
	\end{align}
	More generally, for $k$ scalars $x_1,\dots,x_k \in \RR$, 
	\begin{align}
	\overset{k}{\underset{n=1}{\boxplus}} x_n &\teq 2\tanh^{-1}\left( {\textstyle \prod_{n=1}^k} \tanh\frac{x_n}{2} \right). \label{eq:bsum} 
	\end{align}
Then the update \eqref{eq:v_last} can be computed by the $\tanh$ rule,
	\begin{equation} \label{eq:v_by_tanh}
	\Gamma_n^{(i)} = \Lambda_n^{(i)} + \sum_{m\in\sM(n) \atop  \langle \zeta^i, H_{mn} \rangle=1  } (-1)^{z_m} \left( \underset{n'\in\sN(m)\setminus n}{\boxplus} \lambda_{H_{mn'}}(\Lambda_{n'}) \right),
	\end{equation}
which completes the computation of the first iteration.
We show that \eqref{eq:v_last} and \eqref{eq:v_by_tanh} are equivalent in Appendix~\ref{sec:bp_rule}.

To iteratively update $\Gamma_n^{(i)}$, BP performs the message passing \cite{Pea88,WLK95} by exchanging two types of messages:
\emph{variable-to-check messages} $\lambda_{H_{mn}}(\Gamma_{n\to m})$ and \emph{check-to-variable messages} $\Delta_{m\to n}$, 
where $\Gamma_{n\to m}=\Lambda_n$ for the initialization.
The proposed LLR-BP decoding algorithm is shown in Algorithm~\ref{alg:LLR-BP}. 
Note that in the conventional BP, if $H$ is over $\GF(q^2)$, each message is a vector of length~$q^2$ \cite{DM98,MD01,DF07}, \cite{PC08,Wan+12,Bab+15},
but we need only scalar messages in our refined algorithm.
We remark that in \eqref{eq:vert} and \eqref{eq:hard} the summation is restricted to the anticommute part. 
	This is nontrivial and is different from the conventional approach.

	\begin{algorithm}
		\begin{flushleft}
			\caption{: LLR-BP for decoding quantum codes over $\GF(q=2^l)$ with binary syndrome} \label{alg:LLR-BP}

			{\bf Input}: 
			$H\in\GF(q^2)^{M\times N}$,  $z\in\{0,1\}^M$, $T_{\max}\in\mathbb{Z}_+$, and   LLR vectors $\{\Lambda_n\in\RR^{q^2-1}\}_{n=1}^N$.
			
			{\bf Initialization}: For $n=1$ to $N$ and $m\in\sM(n)$, let 
			$$\Gamma_{n\to m} = \Lambda_n$$
			\begin{itemize}
				\item[] and   compute $\lambda_{H_{mn}}(\Gamma_{n\to m})$.
			\end{itemize}

			{\bf Horizontal step}: For $m=1$ to $M$ and $n\in\sN(m)$, compute 
			\begin{equation} \label{eq:horiz}
			\Delta_{m\to n} = (-1)^{z_m} \underset{n'\in\sN(m)\setminus n}{\boxplus} \lambda_{H_{mn'}}(\Gamma_{n'\to m}).
			\end{equation}
			
			{\bf Vertical step}: For $n=1$ to $N$ and $m\in\sM(n)$, compute 
			\begin{equation} \label{eq:vert}
			\Gamma_{n\to m}^{(i)} = \Lambda_n^{(i)} + \sum_{m'\in\sM(n)\setminus m \atop \langle \zeta^i, H_{m'n} \rangle=1} \Delta_{m'\to n}, 
			~~\forall~~ i=0,\dots,q^2-2,
			\end{equation}
			\begin{itemize}
				\item[] and compute  $\lambda_{H_{mn}}(\Gamma_{n\to m})$.
			\end{itemize}

			{\bf Hard-decision step}: For $n=1$ to $N$, compute 
			\begin{equation} \label{eq:hard}
			\Gamma_n^{(i)} = \Lambda_n^{(i)} + \sum_{m\in\sM(n) \atop \langle \zeta^i, H_{mn} \rangle=1} \Delta_{m\to n}, ~~\forall~~ i=0,\dots,q^2-2.
			\end{equation}
			\begin{itemize}
				\item Let $\hat e = (\hat e_1, \dots, \hat e_N)$, where 
				$\hat e_n=0$, if $\Gamma_n^{(i)} > 0$ for $i=0,\dots,q^2-2$, and
				$\hat e_n=\argmin\limits_{\zeta^i\in\GF(q^2)} \Gamma_n^{(i)}$, otherwise.
			\end{itemize}
			
			The horizontal, vertical, and hard-decision steps are iterated until that 
				the inferred $\hat e$ has syndrome $z$ (i.e., $\langle \hat e, H\rangle = z$)
				or that the maximum number of iterations $T_{\max}$  is reached.

		\end{flushleft}
	\end{algorithm}

We explain the messages. 
First, write \eqref{eq:upd_n} as
\begin{align*}	
\langle e_n, H_{mn} \rangle = z_m + \textstyle \sum_{n'\in\sN(m)\setminus n} \langle e_{n'}, H_{mn'} \rangle.  
\end{align*}
This suggests that the information from a neighboring check~$m$,  
together with the syndrome bit $z_m$, will tell us the likelihood of the error $e_n$ commuting with $H_{mn}$ or not, which can be quantified by a scalar. 
Thus $\lambda_{H_{mn}}(\Gamma_{n\to m})$ is the (scalar) message that variable node $n$ has to send to check node $m$, where $\Gamma_{n\to m}$ initialized to $\Lambda_n$. 
Then check node~$m$ combines the incoming $\lambda_{H_{mn'}}(\Gamma_{n'\to m})$ and generates $\Delta_{m\to n}$ \eqref{eq:horiz}, which is the (scalar) message that check node $m$ has to send to variable node $n$.

Consequently,  variable node $n$ collects the messages $\Delta_{m\to n}$  for $m\in\sM(n)$, 
	together with the initial belief $\Lambda_n$, to update the running belief $\Gamma_{n}$ \eqref{eq:hard}. 
Updating $\Gamma_{n\to m}$ \eqref{eq:vert} can be simplified by, for each entry $i$, 
$\Gamma_{n\to m}^{(i)}=  \Gamma_{n}^{(i)}-\Delta_{m\to n}$, if $\langle\zeta^i,H_{mn}\rangle=1$, and 
$\Gamma_{n\to m}^{(i)}=  \Gamma_{n}^{(i)}$, otherwise.
Then the updated $\Gamma_{n\to m}$ is used to update the message $\lambda_{H_{mn}}(\Gamma_{n\to m})$, which will be sent to check node $m$ for the next iteration.
(The computations may be further simplified as in Remark~\ref{rmk:cmpx0}.)
The two types of message passing (\mbox{$m$-to-$n$} and \mbox{$n$-to-$m$}) are iterated until a stop criterion is achieved.

Note that the estimate of $\Gamma_{n}^{(i)}$ by BP is very good if the neighboring messages incoming to a node are nearly independent \cite[Sec~4.2]{Gal63} 
(e.g., when  the parity checks have small overlap).  
In particular, if the Tanner graph is a tree, BP is exact and converges quickly.
Quantum codes inevitably have correlated  {(dependent)} messages  {due to short cycles}.
We will discuss more in the section of simulations (Sec.~\ref{sec:sim}).

\subsection{Some remarks}

\begin{remark} \label{rmk:tbl}
	The computations in \eqref{eq:la}--\eqref{eq:bsum} can be efficiently computed by numerical or lookup-table methods~\cite{Gal63}, such as the Jacobian logarithm \cite{RVH95,HOP96,HEAD01,WSM04}.
		The function $\lambda_\eta(\cdot)$ can be computed by repeatedly using a function ${ f:\RR^2\to \RR }$ defined by
		$ f(x,y) \triangleq \ln(e^x+e^y) = \max(x,y) +  \ln(1+e^{-|x-y|}) $.
		Since  $\ln(1+e^{-|x-y|})\in [0, \ln(2)] \subset [0, 0.69315]$, it can be implemented by a lookup-table or any numerical methods.
		The multiplications can be avoided if the lookup-table method is used.
		As in \cite{HOP96}, $x\boxplus y = \ln\frac{1+e^{x+y}}{e^x+x^y}$ $= \ln(1+e^{x+y}) - \ln(e^x+x^y)$, which can also be computed by~$f$.
\end{remark}

\begin{remark} \label{rmk:eq}
	When $q=2$, the LLR-BP in Algorithm~\ref{alg:LLR-BP} is equivalent to the refined linear-domain BP for binary quantum codes \cite[Algorithm~3]{KL20}. 
	The linear-domain algorithm can also be extended to $q=2^{l}$. 
	A linear-domain algorithm is suitable for a decoder with fast floating-point multiplication.
\end{remark}

\begin{remark} \label{rmk:sch}
	Algorithm~\ref{alg:LLR-BP} is specified using a \emph{parallel schedule}. Other schedules can also be used (e.g.,  a \emph{serial schedule} is used in \cite[Algorithms~2 and~4]{KL20}). 
\end{remark}

\begin{remark} \label{rmk:cmpx0}
	The computations in Algorithm~\ref{alg:LLR-BP} may be simplified. First, compute \eqref{eq:hard} but do not compute \eqref{eq:vert}.
	Then it is not hard to show that \eqref{eq:vert} can be omitted and
	\begin{equation} \label{eq:simp}
	\lambda_{H_{mn}}(\Gamma_{n\to m}) = \lambda_{H_{mn}}(\Gamma_n) - \Delta_{m\to n}.
	\end{equation}	
	Similarly, to compute each $\Delta_{m\to n}$ in \eqref{eq:horiz}, one can first compute 
	$\Omega_m \teq \underset{n\in\sN(m)}{\boxplus} \lambda_{H_{mn}}(\Gamma_{n\to m})$ 
	and then compute 
	$\Delta_{m\to n} = \Omega_m \boxminus \lambda_{H_{mn}}(\Gamma_{n\to m})$,\footnote{
		This may be an approximation since, 
			before computing $\Omega_m$,
			a tiny disturbance needs to be introduced to $\lambda_{H_{mn}}(\Gamma_{n\to m})$ if it is zero.
		Then when using $x\boxminus y$, it can be guaranteed that $y\ne 0$ and $|x|<|y|$.
		This does not affect the performance in simulations if the disturbance is small enough. 
		}
	where 
	\[
	x \boxminus y  = 2\tanh^{-1}\left( \tanh\frac{x}{2} \,\,/\, \tanh\frac{y}{2} \right) 
	\]
	for $x,y\in\RR$ such that $y\ne 0$ and $|x|<|y|$.
\end{remark}

\begin{remark} \label{rmk:cmpx_c}
	It is known that using a conventional BP over $\GF(q^2)$ has a high check-node complexity $O(q^2\log q^2)$ per edge, which 
	dominates the \mbox{complexity of BP \cite{MD01,DF07,Bab+15}}. 
	Since Algorithm~\ref{alg:LLR-BP} uses scalar messages, the check-node complexity is $O(1)$ per edge, independent of $q$.
	Beside the low complexity, using scalar messages has advantages in performance as explained in the following.
\end{remark}

Since only scalar messages are exchanged in Algorithm~\ref{alg:LLR-BP}, it is straightforward to apply the techniques of \emph{message normalization} or \emph{message offset} to improve the performance of BP when the messages are overestimated \cite{CF02b,CDE+05,YHB04}. 
\begin{itemize}
	\item Message normalization by  $\alpha_c>0$: 
	the message $\Delta_{m\to n}$ in \eqref{eq:horiz} is normalized to $\Delta_{m\to n}/\alpha_c$ prior to  the subsequent computations of the algorithm.
	\item Message offset by   $\beta>0$: 
	the message $\Delta_{m\to n}$ in~\eqref{eq:horiz} is offset to $\sign(\Delta_{m\to n})\times \max(0,|\Delta_{m\to n}|-\beta)$  prior to  the subsequent computations of the algorithm.
\end{itemize}
Messages are overestimated  because  the overlap between two parity checks is not small, causing (unreliable) dependent messages passing in the Tanner graph.
Quantum stabilizer codes  inevitably have this issue and can suffer significant BP performance loss or have high error-floors \cite{MMM04,PC08,Wan+12,Bab+15,ROJ19,PK19,LP19,KL20,RWBC20}.

\begin{remark} \label{rmk:perf}
	In the linear-domain algorithm \cite{KL20},
	it needs to compute $(\cdot)^{1/\alpha}$ for message normalization or $(\cdot)\times 1/e^\beta$ for message offset. 
	Applying message normalization or offset in LLR-BP is much simpler as shown above.
\end{remark}

\begin{remark} \label{rmk:cmpx2}
	For the implementation cost,
	a practical concern   is that, compared to a linear-domain algorithm, 
	an LLR algorithm can be implemented with smaller bit-width  (e.g., each scalar is represented by the six most significant bits (MSBs) for decoding classical binary codes as suggested by Gallager \cite{Gal63}).
	This will be discussed more in the next section.
	Since $\lambda_\eta(\cdot)$ and the operation $\boxplus$ can be approximately computed using looking-up tables (Remark~\ref{rmk:tbl}),
	the LLR algorithm with message normalization or offset does not need multiplications.
	(If $\alpha_c$ is used, choose $1/\alpha_c$ to be a value with smaller bit-width and then the computation of $(\cdot)\times 1/\alpha_c$ only takes several bitwise-shifts and additions.)
\end{remark}

 \begin{remark} \label{rmk:add_vs_lin}
 	Suppose that  $H$ is constructed from $\tilde H$ by the  CSS extension (\ref{eq:css_w}).
 	If a binary syndrome $z\in\{0,1\}^M$ is obtained according to $H$,
 	a nonbinary syndrome $\tilde z\in\GF(2^{2l})^{\frac{M}{2l}}$ can be derived from $z\in\{0,1\}^M$.\footnote{
 		See \cite[Table~5]{Bab+15} for the case of $q=2$; this can be generalized to  $q=2^l$. 
 		}
 	Then the classical nonbinary~BP (referred to as \emph{CBP$_{q^2}$}) can decode the syndrome $\tilde z$ according to $\tilde H$. (See Appendix~\ref{sec:CBP}.)
	
	On the other hand, given a parity-check matrix $\tilde H\in \GF(2^l)^{M'\times N}$ of a classical linear code over $\GF(2^l)$ with short cycles, 
	$\tilde H$ can be extended as an additive check matrix $H\in \GF(2^l)^{lM'\times N}$ by a step like \eqref{eq:css_w}. 
	Then we can use a decoding strategy like Algorithm~\ref{alg:LLR-BP} (possibly with $\alpha_c$ or~$\beta$) to improve the BP performance on this classical code.
\end{remark}
 
 	Though $(H,z)$ and $(\tilde H, \tilde z)$ in Remark~\ref{rmk:add_vs_lin} contain the same amount of information,
 	the results of using Algorithm~\ref{alg:LLR-BP} on $(H,z)$ and using CBP$_{q^2}$ on $(\tilde H, \tilde z)$ would be different.
 	The major difference is that CBP$_{q^2}$ is more likely to have overestimated messages when there are short cycles.
 	This is a disadvantage of using CBP$_{q^2}$ even if the quantum code is obtained from a classical linear code.
 	For comparison, we provide a detailed discussion in Appendix~\ref{sec:CBP}  using Steane's code \cite{Steane96}.
 	In the comparison, Algorithm~\ref{alg:LLR-BP} indeed handles the  message-overestimate problem better.
	In the following simulations, we will focus on the performance of Algorithm~\ref{alg:LLR-BP}.

\section{Simulation results} \label{sec:sim}

Several advantages of using the  scalar-based LLR-BP algorithm are suggested in the above remarks.
In this section, we provide numerical results to show these advantages, focusing on how good performance can be achieved with low complexity.
{Algorithm~\ref{alg:LLR-BP} supports $q=2^l$, so it will be referred to as \mbox{\emph{LLR-BP$_{q^2}$}}, depending on the $q$ used.}
For $q=2$ \mbox{($q^2=4$)}, Algorithm~\ref{alg:LLR-BP} is referred to as \mbox{\emph{LLR-BP$_4$}}. 
Its linear-domain analogue \cite[Algorithm~3]{KL20} is referred to as \mbox{\emph{linear-BP$_4$}}.

In decoding, the number of short cycles is positively correlated with the (mean) \emph{row-density} of the check matrix,
$ \kappa = \frac{1}{M}\sum_{m=1}^M |H_m|/N $. 
We will consider codes with $\kappa$ from small to large. 
The first is a $[[129,28]]_2$ hypergraph-product (HP) code with $\kappa \approx 5.9/129 \approx 0.0457$, 
the second is a $[[256,32]]_2$ bicycle code with $\kappa = 16/256 = 0.0625$, 
and the third is a $[[126,28]]_2$ generalized bicycle (GB) code with $\kappa = 10/126 = 0.0794$.
Then we construct $[[256,32]]_{q=4}$ and $[[126,28]]_{q=4}$ codes by using \eqref{eq:css_w} and \eqref{eq:css_g} from the second and third codes, respectively.
Note that \eqref{eq:css_w} or \eqref{eq:css_g} will result in much more short cycles.

It is known that the BP performance on quantum codes may be improved by using a serial schedule \cite{KL20,RV20,KCL21}. 
A fully parallel implementation of message passing is preferred for faster decoding. 
Herein, we demonstrate \mbox{LLR-BP$_{q^2}$} for $q=2$ or $4$ with the parallel schedule and try to improve the decoding performance.

A scalar can be represented as a floating-point number $$(-1)^{b_0}\times 1.(b_1b_2\dots b_{k-1})_2 \times 2^{\rm exp}$$
as in the IEEE 754 standard \cite{IEEE754}, where $b_0,b_1,\dots,b_{k-1}\in\{0,1\}$ 
and $k$ is called the \emph{bit-width}.
Using a larger~$k$ increases the precision (for both additions or multiplications) but also increases the physical hardware area and computation time.
We use the  $[[129,28]]$ HP code to show that LLR-BP requires a smaller $k$ compared to linear-domain BP (see Sec.~\ref{sec:sim_129}).
This is consistent with the classical case, despite that the HP code has many short cycles.

Next, we describe the details of the simulations.
We will evaluate the performance of various decoding setups and consider the generalized depolarizing errors (Definition~\ref{def:ch}) on certain quantum stabilizer codes. 
Each initial LLR vector $\Lambda_n$ in $\{\Lambda_n\}_{n=1}^N$  is  set to 
	\begin{equation} \label{eq:init}
	\textstyle \Lambda_n = \left( \ln\frac{(1-\epsilon_0)}{\epsilon_0/(q^2-1)}, \dots,\ln\frac{(1-\epsilon_0)}{\epsilon_0/(q^2-1)} \right) \in \RR^{q^2-1},
	\end{equation}
where $\epsilon_0$ is chosen to be the depolarizing error rate $\epsilon$ or a certain constant independent of $\epsilon$. 
The reasons to choose $\epsilon_0$ a constant are related to the performance~\cite{HFI12}, as well as the complexity, and will be seen later.  
Using a constant $\epsilon_0$ also avoids the need to probe or estimate the channel statistic $\epsilon$.

For each simulation of a data point, we try to collect at least 100 logical errors, i.e., $\hat e - e \notin \row(H)$. 
Let $C = \row(H)$. 
A logical error occurs when the syndrome is falsely matched (${ \hat e - e \in C^\perp\setminus C }$) 
or $T_{\max}$ is reached but the syndrome is mismatched (${ \hat e - e \notin C^\perp }$).
We will try to minimize the maximum number of iterations $T_{\max}$ or match it to the literature.

For comparison, consider a \emph{generalized} bounded-distance decoding (BDD) with a lookup-table to have logical error rate
	\begin{equation} \label{eq:BDD}
	P_\text{e,BDD}(N,t,\gamma) = 1- \textstyle \left(\sum_{j=0}^t \gamma_j\binom{N}{j} \epsilon^j (1-\epsilon)^{N-j} \right),
	\end{equation}
where $t$ is the (error-)correction radius and $\gamma_j$ is the percentage of weight-$j$ errors to be corrected.
We  simply denote it by $P_\text{e,BDD}(N,t)$ if ${\gamma_j=100\%}$ for all $j\le t$.
Also, we  redraw the performance curves given in the literature. 
If it is based on independent $X$--$Z$ channel with a cross probability $\epsilon_b$, 
it will be rescaled according to a conversion rule in \cite{MMM04}:
	\begin{equation} \label{eq:b2e}
	\epsilon = 3\epsilon_b/2. 
	\end{equation}
Using this rule may slightly overestimate the performance in the depolarizing channel but it is acceptable.

In the following, we may omit $q$ when referring to an $ [[N,K,D]]_{q=2} $ code.
All tested codes are non-degenerate, and for degenerate codes the readers can refer to \cite{KL21}.

\subsection{$[[129,28,3]]$ hypergraph-product  code} \label{sec:sim_129}
First, we consider a $[[129,28,3]]$ HP code \cite{TZ14}, which is constructed with the $[7,4,3]$ and $[15,7,5]$ BCH codes, discussed in \cite{LP19,KL20,KCL21}.
The performance of BP decoding on this code  is bad  with the parallel schedule.
It can be improved by using a serial schedule \cite{KL20,KCL21} or using a neural-network based BP$_2$ (NBP) \cite{LP19}.
Herein, we show that it can be improved by using message offset with  low complexity, while keeping the parallel schedule.

\begin{figure} 
	\centering \includegraphics[width=0.52\textwidth]{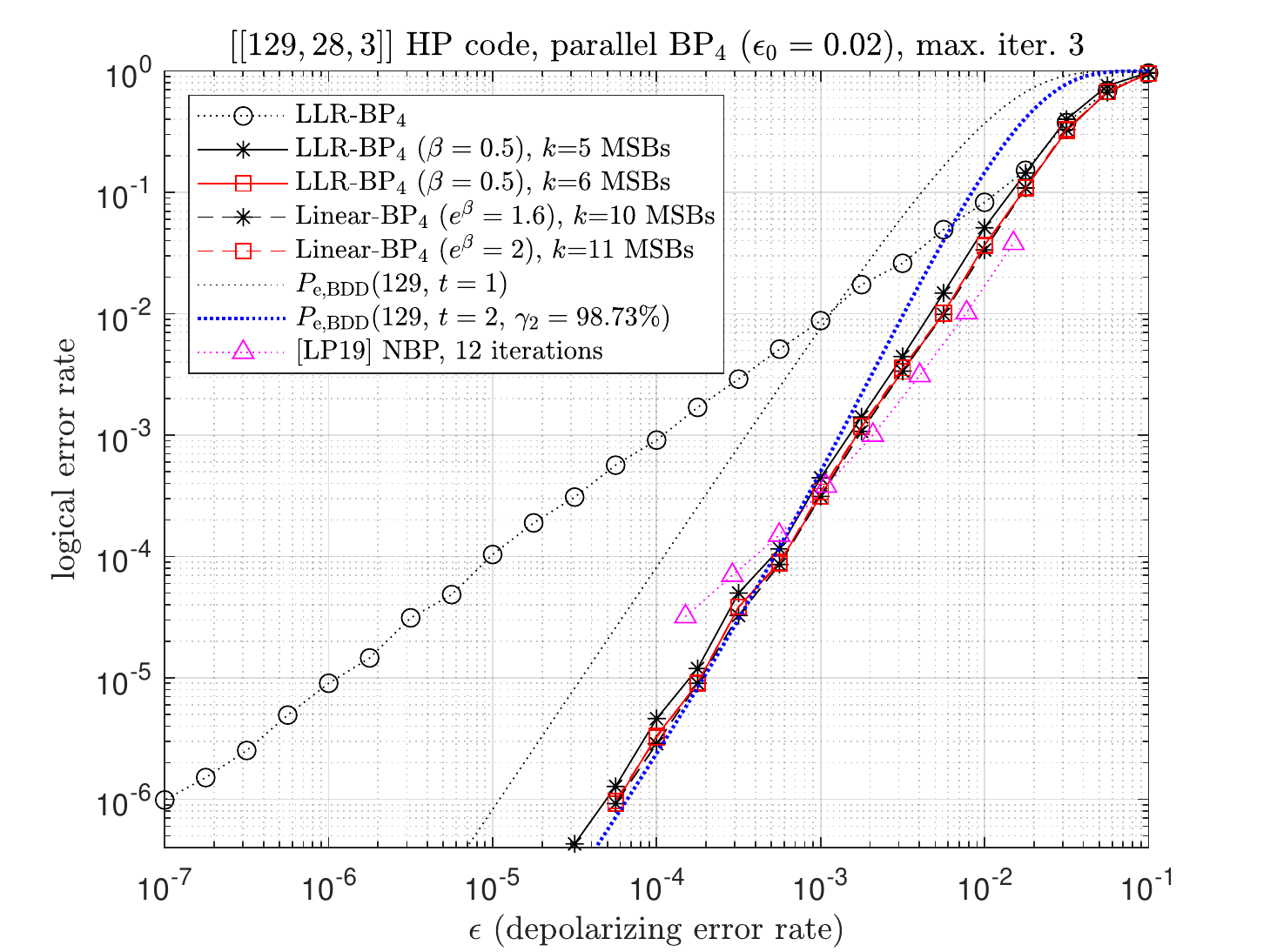}
	\caption{Performance curves of various decoders on the $[[129,28,3]]$ HP code. The [LP19] curve  is converted from \cite{LP19}, using \eqref{eq:b2e}.
	} \label{fig:hp}
\end{figure}

If a correction radius of $t=2$ is considered, this code is known to have $\gamma_0=1$, $\gamma_1=1$ and $\gamma_2\approx 98.73\%$ \cite{KCL21}.
We plot the performance curves of \mbox{LLR-BP$_4$}, \mbox{linear-BP$_4$}, NBP, and the lookup-table decoders $P_\text{e,BDD}(N,t=1)$ and $P_\text{e,BDD}(N,t=2,\gamma_2=98.73\%)$ in Fig.~\ref{fig:hp}. 
Due to the many short cycles in the Tanner graph, LLR-BP$_4$ does not perform well. 
However, after applying a message offset with $\beta=0.5$, its performance is close to the lookup-table decoder.
We also do simulations by truncating all messages to some bit-width $k$.
The required bit-width is about $k=5$ or $6$,  which is close to Gallager's expectation in Remark~\ref{rmk:cmpx2}, though sometimes larger bit-width may be needed due to the short cycles.

When initializing $\Lambda_n$ \eqref{eq:init}, we use a fixed ${\epsilon_0=0.02}$ (also truncated to $k$ MSBs). 
Otherwise,  if ${\epsilon_0=\epsilon}$, then $|\Lambda_n^{(i)}|$ becomes too-large when $\epsilon$ gets small: 
it will be insensitive to small $|\Delta_{m\to n}|$ when performing~\eqref{eq:vert} or~\eqref{eq:hard}, causing ineffective message passing.
In the simulations, it only needs a maximum number of iterations $T_{\max}=3$, so the complexity is very low.
(The bad performance of LLR-BP$_4$ without message offset is irrelevant to constant $\epsilon_0$ or $T_{\max}=3$.)
On the other hand, if linear-BP$_4$ is used, then it requires bit-width $k=11$ based on the choice of $e^\beta = 2$ for message offset, which can be implemented by bitwise-shift.
	In linear domain, one may also consider  $e^\beta = 1.6$: this only needs $k=10$ to achieve the same performance but implementing $(\cdot)\times 1/e^\beta \approx (\cdot)\times 0.625$ (Remark~\ref{rmk:perf}) needs a multiplication or two bitwise-shifts with one addition.

We choose an offset $\beta$ (which is like a threshold) and a constant $\epsilon_0$ \eqref{eq:init} 
by the following criteria, together with some pre-simulations like \cite[Fig.~11]{KL20}. 
We need $|\Delta_{m\to n}|$ large enough to pass the threshold~$\beta$.
The magnitude $|\Delta_{m\to n}|$ would increase with the number of iterations \cite{Gal63}.
Since we use a small ${T_{\max}=3}$ for this code, $\beta$ cannot be too large. 
(${\beta=0.5}$ here is comparatively small than ${\beta=2.75}$ for the other cases with larger $T_{\max}$ that will be discussed later.) 
Also, we need small enough $\epsilon_0$ to have large enough $|\Lambda_n|$, so that $|\Delta_{m\to n}|$ is large enough to pass the threshold $\beta$. 
Thus we have large $|\Lambda_n|$ but relatively smaller $|\Delta_{m\to n}|$ due to small ${T_{\max}=3}$.
Consequently, $k$ should be large enough as discussed in the last paragraph. 
(Using a $k$ smaller than the lower bounds suggested in Fig.~\ref{fig:hp} can cause large performance degradation, due to ineffective message passing.)

On the other hand, using $\alpha_c$ (which will be used later) does not require a fixed $\epsilon_0$ since there is no threshold effect of $\beta$.
There is a certain systematic way to choose the normalization value \cite[Sec.~III-B and Fig.~3]{KL21}.
To choose a better normalization value, some pre-simulations are useful \cite[Fig.~9]{KL20}.

Finally, we use \eqref{eq:b2e} to redraw the NBP performance curve given in \cite{LP19}. (The NBP is also efficient in run-time, but 12 iterations are needed.) 
The curve shows that using a trained neural-network decoder is able to have a better performance for large~$\epsilon$ but is hard to achieve a low error-floor for small~$\epsilon$.

\subsection{Binary quantum bicycle codes}

Second, we simulate bicycle codes \cite{MMM04} and GB codes \cite{KP13} as in \cite[Fig.~5]{PK19}.
We construct the $[[126,28,8]]$ GB code defined in \cite{PK19} (to be decoded with ${T_{\max}=32}$) and a $[[256,32]]$ random  bicycle code discussed in \cite{KL20} (to be decoded with ${T_{\max}=12}$).
The results are shown in Fig.~\ref{fig:bic}.

For reference, the curve denoted [PK19] is the result for the GB code given in \cite{PK19}, which is based on a layered (serial) schedule. 
We show that BP with parallel schedule works as well on this GB code. 
It can be further improved by using $\alpha_c$. (Using $\beta$ with proper $\epsilon_0$ also works.)
Gallager estimated that BP would have  performance close to ${ P_\text{e,BDD}(N,\,t\approx d) }$ \cite{Gal63}, 
much better than the typical BDD performance $ P_\text{e,BDD}(N,\,t\approx d/2) $. 
We draw the case $ P_\text{e,BDD}(N,\,t\approx d) $ for reference. 
The simulation results agree with Gallager's expectation.

\begin{figure} 
	\centering \includegraphics[width=0.52\textwidth]{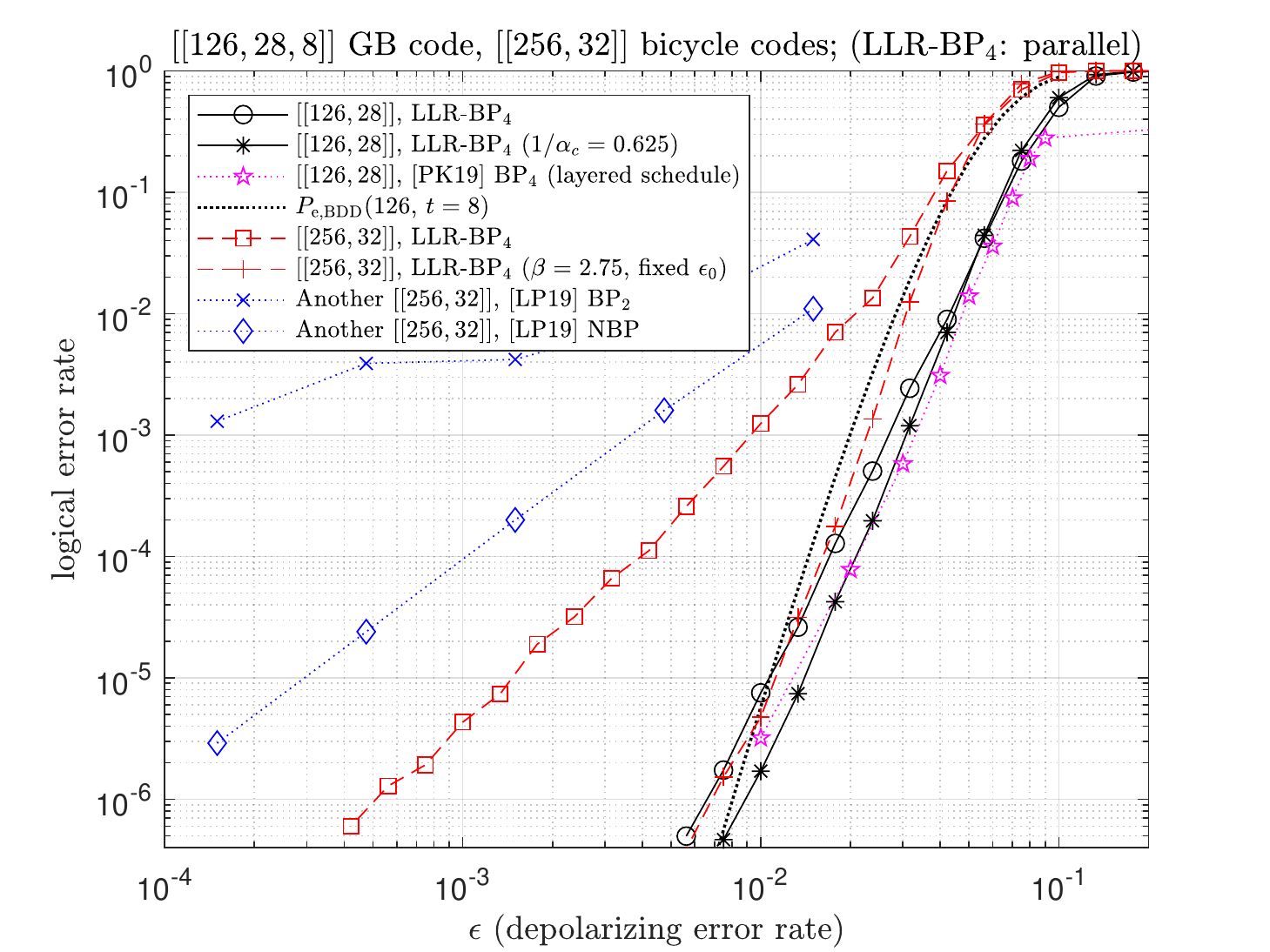}
	\caption{Performance curves of various decoders on the $[[126,28,8]]$ GB code and the $[[256,32]]$ random bicycle codes. 
		The [PK19] curve  is from \cite{PK19}.
		The [LP19] curves  are converted from \cite{LP19},  using \eqref{eq:b2e}.
	} \label{fig:bic}
\end{figure}

For the $[[256,32]]$ (random) bicycle code, it has a high error-floor \cite{KL20} (due to random construction~\cite{KL20b}). 
We show that using ${\beta=2.75}$ with ${\epsilon_0=5\times10^{-3}}$ significantly improves the error-floor. (Using $\alpha_c$ also works.)
The code's check matrix has a constant row-weight 16, and thus the code distance $D\le 16$, due to a row-deletion step in the bicycle construction \cite{MMM04}.
(Increasing the row-weight may lead to some unwanted side-effects in BP performance such as delayed waterfall roll-off\cite[Fig.~6]{MMM04}.)
If $D\approx 16$, this code may have BP performance approaching ${P_\text{e,BDD}(N,\,t=16)}$ by Gallager's expectation. 
Increasing $T_{\max}$ can achieve this performance.
For the convenience of discussion, we will show this for the case of $q=4$ in the next subsection.
(The qudit-wise error correction capability would be  at least the same as  the case of $q=2$ by using the CSS extension \eqref{eq:css_w}.)

The  performance curves (labeled [LP19] BP$_2$ and NBP) of another random $[[256,32]]$ bicycle code in \cite{LP19}
are  rescaled by \eqref{eq:b2e} and also plotted in Fig.~\ref{fig:bic} for comparison. 
Although NBP improves from BP$_2$ by several orders of magnitude, the resultant performance is not good enough because the original error-floor of [LP19] BP$_2$ is too-high.
From our experience, this is likely because the binary generator vector used for construction has many consecutive ones, causing too-many short cycles in the Tanner graph.
Using the random bicycle construction appropriately,\footnote{
	Our $[[256,32]]$ bicycle code is constructed by a binary generator vector with ones at bits 
	$1,     3,     9,    59,    68,    69,   107,   112$; 
	and in the row-deletion step, rows 
	$1,     2,    12,    59,    60,    68,    70,    73,    74,    76,    91,    92,   100,   115,   117,   120$ are deleted.
	For this code size, a purely-random construction is usually fine to achieve a good performance by proper message normalization or offset;
		it is more important to make sure that there are no more than three consecutive ones in the generator vector.
		The row-deletion  becomes tricky (to prevent  too irregular Tanner graphs) only when $N$ is large (e.g., $N>3000$) \cite{KL20b}.
	}
the curves we obtained have much lower error-floors.

\subsection{Nonbinary quantum bicycle codes}

Herein, we extend the previous $[[256,32]]$ bicycle code and $[[126,28]]$ GB code to $q=4$ by \eqref{eq:css_w} and \eqref{eq:css_g}, respectively.

The previous $[[256,32]]$ code has a check matrix $H = \left[ \begin{smallmatrix} \tilde H \\ \omega_0\tilde H \end{smallmatrix} \right]$,
where $\tilde H$ is a binary matrix such that $\tilde H \tilde H^T=O$ and $\omega_0$ is a primitive element of $\GF(4)$.
By the   CSS extension, we  have the following check matrix for a $[[256,32]]_{4}$ code
	$H = \left[ \begin{smallmatrix} \tilde H \\ \omega\tilde H \\ \omega^2\tilde H\\ \omega^3\tilde H \end{smallmatrix} \right]$,
where $\omega$ is a primitive element of $\GF(16)$.

\begin{figure} 
	\centering \includegraphics[width=0.52\textwidth]{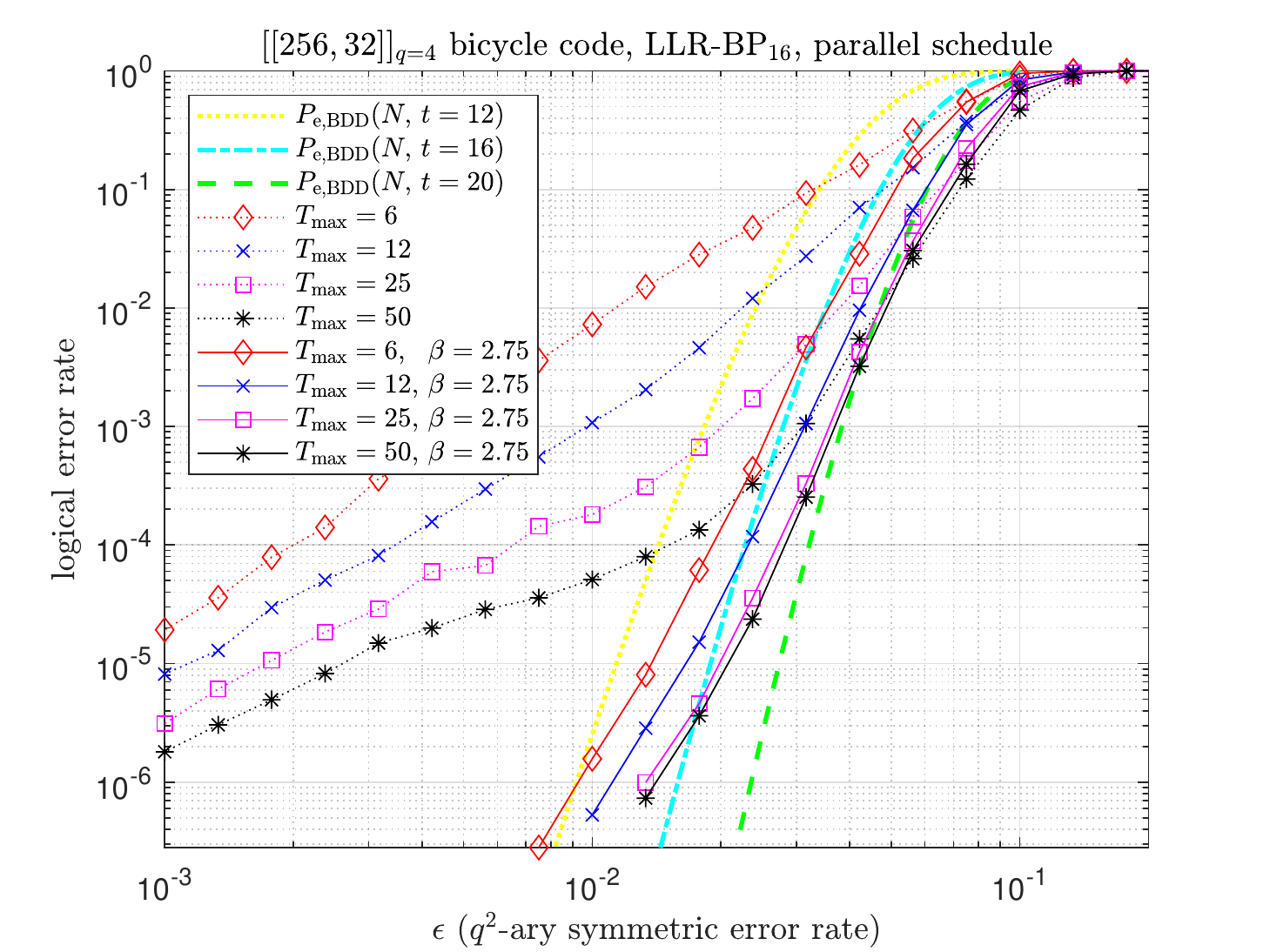}
	\caption{
	Performance of Algorithm~\ref{alg:LLR-BP} on the $[[256,32]]_{q=4}$ random bicycle code. $D$ is unknown but $\le 16$.
	When $\beta$ is applied, a fixed $\epsilon_0=5\times 10^{-3}$ is used for initializing $\Lambda_n$ \eqref{eq:init} 
	to prevent delayed waterfall roll-off.
	}\label{fig:256_q4}
\end{figure}

We perform the decoding of this  $[[256,32]]_{4}$ code in the \mbox{$q^2$-ary} symmetric channel (Def.~\ref{def:ch})
by Algorithm~\ref{alg:LLR-BP} (now LLR-BP$_{16}$). 
%
Different values of $T_{\max}$ are considered, and the results are shown in Fig.~\ref{fig:256_q4}. 
Using either $\alpha_c$ or $\beta$ improves the BP performance (to a similar level), and we show the case of using $\beta$.
Several BDD cases are also provided.
It can be seen that using $\beta$ significantly improves the error-floor performance even with small ${T_{\max}=6}$.
For ${T_{\max}=12}$, the performance trend is similar to the corresponding case in Fig.~\ref{fig:bic}.
By increasing $T_{\max}$ from 12 to 25, the performance improves by about  half an order of magnitude.
With message offset, there is no significant improvement if $T_{\max}$ is further increased to~50.
The CSS extension does not increase $D$, so the code still has $D\le 16$.
With enough $T_{\max}$, Algorithm~\ref{alg:LLR-BP} with message offset may have performance close to $P_\text{e,BDD}(N,\,t=16)$ 
at logical error rate around $10^{-5}$.

We remark that the high error-floor problem can be improved by carefully doing the row-deletion during the bicycle construction \cite{MMM04,KL20b}. 
However, as ${D\le 16}$, the performance limit is still about ${P_\text{e,BDD}(N,\,t=16)}$, and this is roughly achieved by using $\alpha_c$ or $\beta$ even with random construction.
Thus, Algorithm~\ref{alg:LLR-BP} (with proper message normalization or offset) has stable performance regardless of the code construction.
This provides more {code candidates in applications} and more flexibility when designing the stabilizer measurements (rows of $H$) for physical implementations.

\begin{figure} 
	\centering \includegraphics[width=0.52\textwidth]{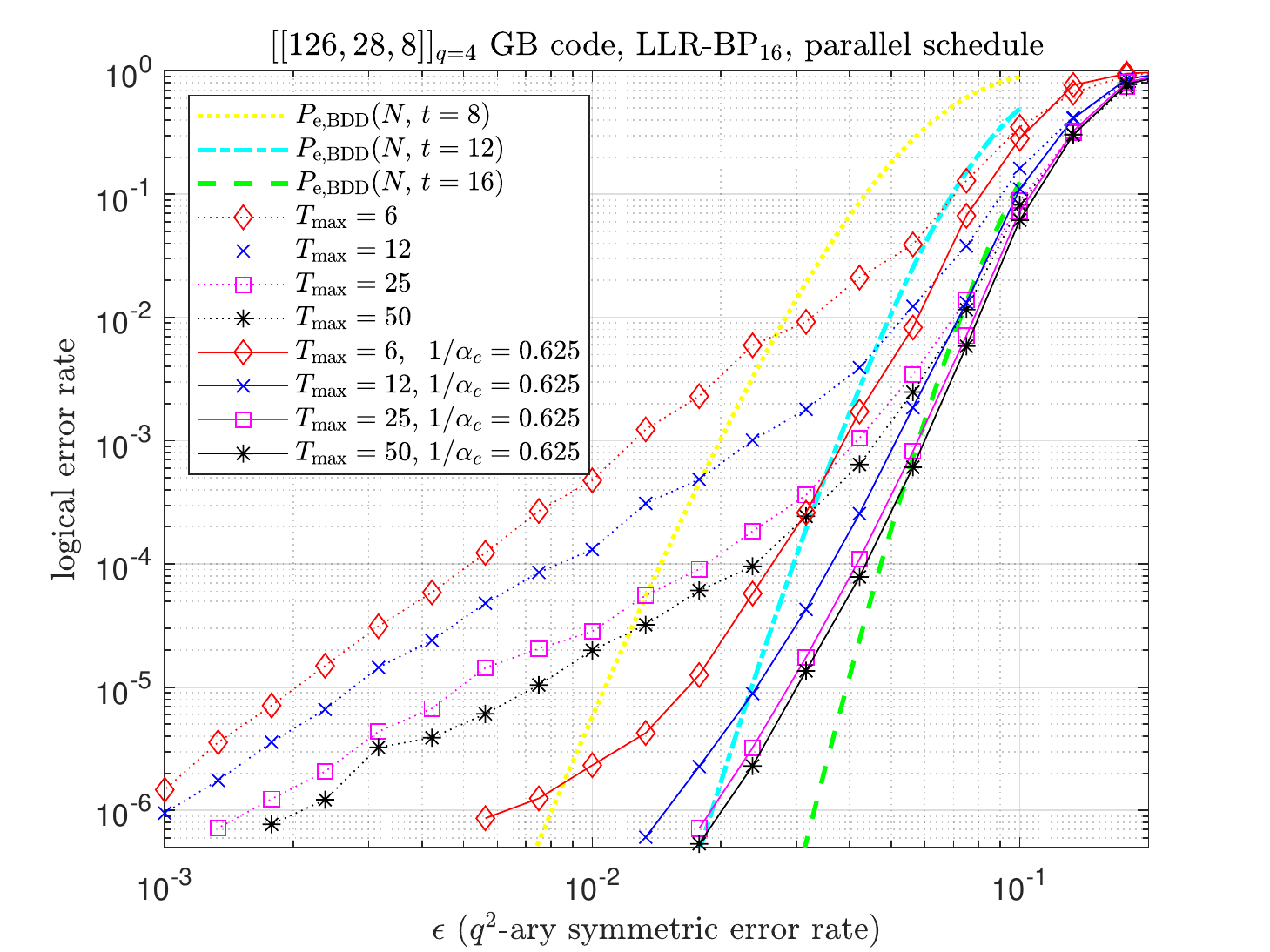}
	\caption{
	Performance of Algorithm~\ref{alg:LLR-BP} on the $[[126,28,8]]_{q=4}$ GB code.
	}\label{fig:126_q4}
\end{figure}

The previous $[[126,28,8]]$ GB code has a binary check matrix 
	${ H \equiv \left[\begin{smallmatrix} H_X&O\\ O&H_Z \end{smallmatrix}\right] }$     with ${H_X\ne H_Z}$, 
and we generalize it as a $[[126,28,8]]_{q=4}$ code by the generalized CSS extension~\eqref{eq:css_g}.
%
The decoding performance of \mbox{LLR-BP$_{16}$}  on this code is shown in Fig.~\ref{fig:126_q4}.
Compared to the corresponding case in Fig.~\ref{fig:bic} ($[[126,28]]$ curves therein), the raw BP does not perform well here. 
Using either $\alpha_c$ or $\beta$ improves the performance (to a similar level).
Here, we demonstrate the case of using $\alpha_c$.
The generalized CSS extension \eqref{eq:css_g} induces more short cycles, causing more overestimated messages and degrading the raw BP performance. 
Applying message normalization significantly improves BP even with small $T_{\max} = 6$.
For large $T_{\max}\ge 25$, using message normalization achieves much better performance close to \mbox{$P_\text{e,BDD}(N,\,t=12)$} at logical error rate $10^{-6}$, 
which means that most errors of weight $\le 12$ are corrected, despite that the minimum distance $D=8$ only.

\section{Conclusion and discussions} \label{sec:con}

We proposed an efficient scalar-based LLR-BP algorithm for decoding quantum codes over $\GF(2^l)$,
which extended our previous work \cite{KL20}.
This is especially useful when many qubits are grouped into a qudit and a nonbinary decoder is used;
our demonstrations   showed that correlated noise can be well handled by the nonbinary BP.

%
The check-node complexity in our scalar-based LLR-BP reduces from $O(2^{2l}\log 2^{2l})$ to $O(1)$ per edge.
This saves the decoding time and keep the quantum coherence better.
Moreover, the LLR algorithm can be implemented with additions and lookup-tables without multiplications. 
It also has other advantages, such as the smaller bit-width and the convenience for message normalization or offset.
Using message normalization or offset to improve the performance of decoding nonbinary quantum codes is very important since they may have more check rows and more short cycles.
Computer simulations were conducted, showing that good performances can be achieved by this low-complexity approach.
The algorithm has stable performance for various quantum codes and constructions. 
%

Our LLR method naturally extends to a joint decoding of data-syndrome errors by BP \cite{KCL21}.

Although our criterion of a successful decoding is $\hat e - e \in \row(H)$, $\hat e$ rarely converges to a degenerate error that is not $e$ when the decoding succeeds in the simulations of this paper. This is because the tested codes are non-degenerate, so BP converges in a way more like classical decoding.

Topological codes are highly degenerate and have stabilizers of low weight around~4 to 6 \cite{Kit03,BM06}
(causing the \emph{classical} minimum distance $\le$ 4 or 6 because the code is self-orthogonal).
Using the proposed BP algorithm on these degenerate codes 
has bad decoding performance.
However, further modification on BP to handle these degenerate codes is possible and this is addressed in \cite{KL21}.\footnote{
	The achieved threshold is roughly 16\% and 17.5\% for decoding surface codes and toric codes, respectively, as shown in \cite[Figs.~14 and~21]{KL21}.
	}

A quantum BCH/RS code \cite{GGB99,GB99,AKS07,LaGua14} has numerous short cycles in its Tanner graph since its parity-check matrix is of high density.
We tested two MDS Hermitian-orthogonal BCH codes $[[17,9,5]]_4$ and $[[17,13,3]]_4$ in \cite{LaGua14} and found that message normalization can improve the scalar-based BP by roughly an order of magnitude.
However, it is still far from any good BDD benchmark due to the many short cycles.
Pre-processing on the parity-check matrix \cite{LALL17,LALL20} or post-processing by ordered statistics decoding (OSD) \cite{JN06,EM06} improve BP on classical BCH/RS codes. 
Our decoder may also be improved through these techniques.

Our approach could be extended to quantum codes over $\GF(p^l)$ for a prime~$p$.
This would reduce the check-node complexity from $O(p^{2l}\log p^{2l})$ to $O(p\log p)$ per edge 
and such nonbinary BP still maintains the correlations between the $p^{2l}$ errors in the error basis.
%
The measurement (trace inner product) outcomes  would be in $\GF(p)$ as in Definition~\ref{def:Tr}.
Given nonzero $ \eta \in \GF(q^2=p^{2l}) $, Corollary~\ref{col:tr_half} should be generalized that the $q^2$ elements of $\GF(q^2)$ are partitioned into $p$ groups (each with $q^2/p$ elements). 
Each BP message is a probability vector over $\GF(p)$ in linear domain, where each entry is the sum of the probabilities of $q^2/p$ elements.
The check-node computation (convolution of two probability vectors over $\GF(p)$) can be done by FFT \cite{MD01,DF07}.

\appendices		

\section{:~ Basic variable-node and check-node update rules} \label{sec:bp_rule}

Herein we show that \eqref{eq:v_last} is equivalent to \eqref{eq:v_by_tanh} by using Def.~\ref{def:la} and  the $\tanh$ rule \eqref{eq:bplus}, \eqref{eq:bsum}. The derivation is similar to the BP$_2$ case \mbox{\cite[Sec.~2.5.2]{RU08}}. 
It suffices to focus on the update for $n=1$ with one row $H_m$. 
We assume that $\sN(m) = \{1,2,\dots,|H_m|\}$ without loss of generality and prove by induction on $|H_m|$.
The strategy is to go through the basic variable-node and check-node update  rules \mbox{\cite[Sec.~V-E]{KFL01}}.

If the row is with $|H_m|=1$, it is trivial that $\Gamma_{n}^{(i)} = \Lambda_{n}^{(i)}$.

If the row is with two nonzero entries, consider
	$$H_m|_{\sN(m)=\{1,2\}}=[\eta,\, \xi].$$ 
For $n=1$, \eqref{eq:v_last} becomes 
	\begin{equation} \label{eq:v_rule}
	\Gamma_1^{(i)} = \Lambda_1^{(i)} + \langle \zeta^i, \eta \rangle (-1)^{z_m} \lambda_{\xi}(\Lambda_{2})
	\end{equation}
by using \eqref{eq:la_1}. 
The result is identical to \eqref{eq:v_by_tanh} for $|H_m|=2$.
Equation~\eqref{eq:v_rule} is called the basic {variable-node update rule}.

Now suppose that the belief of $n=1$ is to be updated by one weight-three row $H_m$ with
\begin{equation*} 
H_m|_{\sN(m)=\{1,2,3\}}=[\eta,\, \xi,\, \xi'].
\end{equation*}
For $\zeta^i\in \GF(q^2)$ such that $\langle \zeta^i, \eta \rangle = 0$, both \eqref{eq:v_last} and \eqref{eq:v_by_tanh} have the same result:
	\begin{align*}
	~~~~\Gamma_1^{(i)} = \Lambda_1^{(i)} \quad \text{(i.e., no update is needed if $\langle \zeta^i, \eta \rangle = 0$).} 
	\end{align*}
For $\zeta^i\in \GF(q^2)$ such that $\langle \zeta^i, \eta \rangle = 1$, if $z_m=0$ in \eqref{eq:v_last},
	\begin{align*} \label{eq:comb}
	&\Gamma_1^{(i)} 
	= \Lambda_1^{(i)}+ \ln\tfrac{ \sum_{\tau,\tau':\langle \tau, \xi \rangle + \langle \tau', \xi' \rangle = 0} P(e_2 = \tau)P(e_3 = \tau') }
								{ \sum_{\tau,\tau':\langle \tau, \xi \rangle + \langle \tau', \xi' \rangle = 1} P(e_2 = \tau)P(e_3 = \tau') }\notag\\
	&= \Lambda_1^{(i)}+ {\ln}\tfrac{ P(\langle e_2,\xi \rangle=0) P(\langle e_3,\xi' \rangle=0) + P(\langle e_2,\xi \rangle=1) P(\langle e_3,\xi' \rangle=1) }
	{ P(\langle e_2,\xi  \rangle=0) P(\langle e_3,\xi' \rangle=1) + P(\langle e_2,\xi \rangle=1) P(\langle e_3,\xi' \rangle=0) },\notag\\
	&\quad \text{\small where the log term equals the LLR of $\langle e_2,\xi \rangle + \langle e_3,\xi' \rangle$ (mod 2),}\notag\\
	&= \Lambda_1^{(i)}+ {\ln}\tfrac{(P(\langle e_2,\xi \rangle=0)/P(\langle e_2,\xi \rangle=1))(P(\langle e_3,\xi' \rangle=0)/P(\langle e_3,\xi' \rangle=1)) + 1 }
	{ P(\langle e_2,\xi  \rangle=0)/P(\langle e_2,\xi \rangle=1) +  P(\langle e_3,\xi' \rangle=0)/P(\langle e_3,\xi' \rangle=1) }\notag\\
	&= \Lambda_1^{(i)} + {\ln}\tfrac{ e^x e^y + 1 }{ e^x + e^y }, \quad \text{where} ~~  x=\lambda_{\xi}(\Lambda_2),\ y=\lambda_{\xi'}(\Lambda_3),\notag\\
	&= \Lambda_1^{(i)} + {\ln}\tfrac{ (e^x+1)(e^y+1) + (e^x-1)(e^y-1) }{ (e^x+1)(e^y+1) - (e^x-1)(e^y-1) }\notag\\
	&= \Lambda_1^{(i)} + 2\tanh^{-1}\left(\tfrac{e^x-1}{e^x+1} \times \tfrac{e^y-1}{e^y+1}\right)\notag\\
	&= \Lambda_1^{(i)} + ( \lambda_{\xi}(\Lambda_2) \boxplus \lambda_{\xi'}(\Lambda_3) ) \quad \text{by   \eqref{eq:bplus}}.
	\end{align*}
In general, we have for  $\langle \zeta^i, \eta \rangle \in\{0,1\}$  and   $z_m\in\{0,1\}$, 
	\begin{equation} \label{eq:c_rule_1}
	\Gamma_1^{(i)} = \Lambda_1^{(i)} + \langle \zeta^i, \eta \rangle (-1)^{z_m} ( \lambda_{\xi}(\Lambda_2) \boxplus \lambda_{\xi'}(\Lambda_3) ),
	\end{equation}
which is identical to \eqref{eq:v_by_tanh} for $|H_m|=3$. 
The combining $(\lambda_{\xi}(\Lambda_2), \lambda_{\xi'}(\Lambda_3)) \mapsto \lambda_{\xi}(\Lambda_2) \boxplus \lambda_{\xi'}(\Lambda_3)$ is called the basic check-node update rule.

Now assume that the belief of $n=1$ is to be updated by one weight-four row $H_m$ with $H_m|_{\sN(m)}=[\eta,\, \xi,\, \xi',\, \xi'']$. 
Consider $\langle e_2,\xi \rangle + \langle e_3,\xi' \rangle$ (mod~2) as a binary random variable;
then its LLR is $\lambda_{\xi}(\Lambda_2) \boxplus \lambda_{\xi'}(\Lambda_3)$ by the above derivation. 
We can further consider a binary random variable 
$\left(\langle e_2,\xi \rangle + \langle e_3,\xi' \rangle\right) + \langle e_4,\xi'' \rangle$ (mod~2); then its LLR is 
$\left( \lambda_{\xi}(\Lambda_2) \boxplus \lambda_{\xi'}(\Lambda_3) \right) \boxplus \lambda_{\xi''}(\Lambda_4)$ by the same reason.
Thus, similar to \eqref{eq:c_rule_1}, 
\begin{align*}
\Gamma_{1}^{(i)} &= \Lambda_{1}^{(i)} + \langle \zeta^i, \eta \rangle (-1)^{z_m} \left( \lambda_{\xi}(\Lambda_2) \boxplus \lambda_{\xi'}(\Lambda_3) \right) \boxplus \lambda_{\xi''}(\Lambda_4) \\
&= \Lambda_{n=1}^{(i)} + \langle \zeta^i, H_{mn} \rangle (-1)^{z_m} \overset{4}{\underset{n'=2}{\boxplus}} \lambda_{H_{mn'}}(\Lambda_{n'}) \text{~ by \eqref{eq:bsum},}
\end{align*}
i.e., \eqref{eq:v_last} and \eqref{eq:v_by_tanh} are equivalent for $|H_m|=4$.
By induction on $|H_m|$ using the same trick and by joining more rows, we have that \eqref{eq:v_last} and \eqref{eq:v_by_tanh} are equivalent.

\newcommand*\refLLRBP{\ref{alg:LLR-BP} }	
\section{:~ Comparison of Algorithm~\protect\refLLRBP and the classical nonbinary BP} \label{sec:CBP}	

Conventionally one would like to use classical nonbinary BP to decode quantum codes since they can be considered as special nonbinary codes.
We will compare this direct use of classical nonbinary BP with our LLR-BP in~Algorithm~\ref{alg:LLR-BP}.

We first describe the decoding of binary quantum codes ${ (q=2) }$. 
Suppose that ${ H\in\GF(4)^{M\times N} }$ and ${ z\in\{0,1\}^M }$ are given.
For convenience, we describe the BP  in linear domain with  
initial beliefs ${ \{p_n=(p_n^I,p_n^X,p_n^Y,p_n^Z)\in\RR^4\}_{n=1}^{N} }$ and 
running beliefs ${ \{q_n=(q_n^I,q_n^X,q_n^Y,q_n^Z)\in\RR^4\}_{n=1}^{N} }$.
Also, the variable-to-check and check-to-variable messages are denoted by 
${ \{q_{nm} = q_{n\to m}\in\RR^4\}_{(m,n):H_{mn}\ne0} }$ and 
${ \{r_{mn} = r_{m\to n}\in\RR^4\}_{(m,n):H_{mn}\ne0} }$, respectively.
Initially,  $q_{nm} = p_n$ for each edge $(m,n)$. 
The following methods are considered in \cite{Wan+12,Bab+15}.

\begin{itemize}
\item[1.] \cite[(13)--(16)]{Wan+12} or \cite[(31)--(34)]{Bab+15}: 
	If the binary quantum code corresponds to an additive code over $\GF(4)$,  {convert the syndrome $z\in\{0,1\}^M$ to a syndrome $\tilde z\in\GF(4)^M$} such that each $z_m=\tr(\tilde{z}_m)$. 
	Thus $\{r_{mn}\}$ can be generated from $\{q_{nm}\}$ using the classical nonbinary BP. 
	Suppose that each ${ r_{mn}=(r_{mn}^{I}, r_{mn}^{X}, r_{mn}^{Y}, r_{mn}^{Z}) }$ is generated, where ${ \sum_W r_{mn}^W = 1 }$.
	According to the properties in Example~\ref{ex:q=2}, do a \emph{post average} as specified in \cite{Wan+12,Bab+15}. 
	For example, if $z_m=0$ and the edge $(m,n)$ is of type~$X$, then the post average rescales the vector as 
	$$\textstyle \tilde r_{mn}=(\frac{r_{mn}^{I}+r_{mn}^{X}}{2}, \frac{r_{mn}^{I}+r_{mn}^{X}}{2}, \frac{r_{mn}^{Y}+r_{mn}^{Z}}{2}, \frac{r_{mn}^{Y}+r_{mn}^{Z}}{2}).$$
	Then $q_n^W = a_n p_n^W \prod_{m\in\sM(n)} \tilde r_{mn}^W$ for all $W$ and $n$, where $a_n$ is a scalar such that $\sum_W q_n^W=1$.
\item[2.] \cite[(44)--(47)]{Bab+15} (only for linear codes): 
	If the binary quantum code corresponds to a linear classical code over $\GF(4)$, 
	assume that the given ${ H=\left[ \tilde H \atop \omega \tilde H \right] }$, 
	where ${ \tilde H \in \GF(4)^{\frac{M}{2}\times N} }$ is a parity-check matrix of the linear code.  
	Then one can treat $M$ syndrome bits ${ \{z_m\in\{0,1\}\}_{m=1}^M }$ as $M/2$ quaternary digits ${ \{\tilde z_m\in\GF(4)\}_{m=1}^{M/2} }$,
	and $\tilde z\in\GF(4)^{M/2}$ is regarded as the syndrome generated by the unknown error and the parity-check matrix $\tilde H$. 
	Using the classical nonbinary BP can decode the syndrome. Note that there is no post average on $r_{mn}=(r_{mn}^{I}, r_{mn}^{X}, r_{mn}^{Y}, r_{mn}^{Z})$, 
	and $q_n^W$ is updated by $q_n^W = a_n p_n^W \prod_{m\in\sM(n)} r_{mn}^W$ for all $W$~and~$n$, where $a_n$ is a scalar such that $\sum_W q_n^W=1$.
\end{itemize}

Both methods can be generalized to ${q=2^l}$, and the check-node complexity is ${O(2^{2l}\log 2^{2l})}$ per edge. 
By a derivation like \eqref{eq:sum_all}--\eqref{eq:v_by_tanh}, it can be shown that method~1, after generalized, is equivalent to Algorithm~\ref{alg:LLR-BP} in the sense that they have the same decoding output if no message normalization or offset is considered. 
However, Algorithm~\ref{alg:LLR-BP} has a check-node complexity $O(1)$ per edge and is more efficient.

If the code is linear, a generalized method~2 for $q=2^l$ is also efficient since the number of checks becomes $\frac{M}{2l}$.
However, method~2 tends to have more biased messages $r_{mn}$, while method~1 and Algorithm~\ref{alg:LLR-BP} tend to have more fair messages $\tilde r_{mn}$. 
Consequently, method~2 may   have overestimated messages. 
For comparison, we generalize method~2 for ${ q=2^l }$.
Let CBP$_{q^2}$ be the classical nonbinary BP (such as \cite{WSM04})  
with input: 
${ \tilde H\in\GF(q^2)^{\frac{M}{2l}\times N} }$,
${ \tilde z\in\GF(q^2)^\frac{M}{2l} }$, 
$T_{\max}\in\ZZ_+$, 
and ${ \{\Lambda_n\in\RR^{q^2-1}\}_{n=1}^N }$;
and with output $\hat{e}\in\GF(q^2)^N$.
%
%
Then we can impose on $\text{CBP}_{q^2}$ with method~2 to solve the decoding problem, as in Algorithm~\ref{alg:CBP}.\footnote{
	$\text{CBP}_{q^2}$ therein is assumed with Hermitian inner product. If it is with Euclidean inner product, simply do conjugate $[\tilde H_{mn}^q]$ before using CBP$_{q^2}$.
	}

	\begin{algorithm}
		\begin{flushleft}
			\caption{: A decoding method for quantum codes over $\GF(q)$ that correspond to classical linear codes over $\GF(q^2)$} \label{alg:CBP}
			
			{\bf Input}:
			$H\in\GF(q^2)^{{M}\times N}$ (where $q=2^l$), 
			$z\in\{0,1\}^{M}$, $T_{\max}\in\mathbb{Z}_+$, 
				$\{\Lambda_n\in\RR^{q^2-1}\}_{n=1}^N$, 
			and a BP oracle $\text{CBP}_{q^2}$.
			
			{\bf Initialization}: 
			\begin{itemize}
				\item 	Derive $\tilde H\in\GF(q^2)^{\frac{M}{2l}\times N}$ from $H$. 
				\item Convert $z$ to $\tilde z\in\GF(q^2)^{\frac{M}{2l}}$ according to the relation between $H$ and $\tilde H$.
			\end{itemize}
		
			\mbox{{\bf Decoding}: run $\hat e = \text{CBP}_{q^2}(\tilde H,\, \tilde z,\, T_{\max},\, \{\Lambda_n\}_{n=1}^N)$.}

		\end{flushleft}
	\end{algorithm}

	The complexity of each algorithm (or method) is as follows.  
	Let $|H|$ be the number of nonzero entries in $H$.
	The variable-node complexity is $ O(q^2) $ per~edge for each algorithm.
	Consider the computation flow as in Remark~\ref{rmk:cmpx0}.
		\begin{itemize}
		\item 
		Algorithm~\ref{alg:LLR-BP}  has complexity $ |H|O(q^2) + |H|O(1) = O(|H|q^2) .$
		\item
		Method~1 (after generalized for nonbinary quantum codes) 
		A generalized method~1
		has complexity $ |H|O(q^2) + |H|O(q^2\log q^2) = O(|H|q^2\log q^2) .$
		\item
		Algorithm~\ref{alg:CBP} has $ |\tilde H| = |H|/2l $, where ${ 2l = \log_2(q^2) }$. 
			Thus, Algorithm~\ref{alg:CBP} 
			has complexity $ |\tilde H|O(q^2) + |\tilde H|O(q^2\log_2(q^2)) = O(|H|q^2) $
			but is only applicable to linear codes.
		\end{itemize}

In \cite{Bab+15}, method~1 is called \emph{standard nonbinary BP} (though with additional post average) 
and method~2 is called \emph{modified nonbinary BP} (though it is  the check matrix  that is modified, not CBP).
The authors in \cite{Bab+15} recommended to use method~2 for linear codes because of fewer short cycles in $\tilde H$.
However, we take issue on method~2 since it may have many overestimated messages to perform less well as shown as follows.

\newcommand*\refCBP{\ref{alg:CBP} }																	
\newcommand*\SEVENq{$[[7,1,3]]_q$ }																	
\subsection*{Comparing Algorithms~\protect\refLLRBP and~\protect\refCBP by using \protect\SEVENq}	

Steane's code \cite{Steane96} is a $[[7,1,3]]_2$ code with a check matrix  $H\in\GF(4)^{6\times 7}$ extended from a classical Hamming code with a binary parity-check matrix $\tilde H\in\GF(2)^{3\times 7}$ by \eqref{eq:css_w}.

We consider to generalize Steane's code by writing ${ \tilde H = \tilde H^{(r)} }$ to mean that the (binary) $\tilde H$ has $r$ rows (checks); 
	then the corresponding ($q^2$-ary) $H = H^{(r)}$ has $2lr$ rows for $q=2^l$ as in \eqref{eq:css_w}.
The (binary) $\tilde H^{(3)}$ can be cyclicly generated by the vector $(1011100)$ with three rows as
	$$ \tilde H^{(3)} = 
	\left[ \begin{smallmatrix}
	1 & 0 & 1 & 1 & 1 & 0 & 0 \\
	0 & 1 & 0 & 1 & 1 & 1 & 0 \\
	0 & 0 & 1 & 0 & 1 & 1 & 1
	\end{smallmatrix}\right],
	$$
which has an irregular Tanner graph.
If we cyclicly generate seven rows, then the resultant (binary) $\tilde H^{(7)}$ has a regular Tanner graph, 
which provides the same error-correction capability for each qudit. 
BP on $\tilde H^{(7)}$ still has overestimated messages but the incorrect dependency may cancel out.

First, consider the case of $q=2$. 
The ($q^2$-ary) $H^{(3)}$ and $H^{(7)}$ have $6$ and $14$ rows, respectively (i.e., Algorithm~\ref{alg:LLR-BP} on $H^{(r)}$ 
encounters more short cycles than CBP$_{q^2}$ on $\tilde H^{(r)}$).
The decoding results of Algorithms~\ref{alg:LLR-BP} and~\ref{alg:CBP} are shown in Fig.~\ref{fig:713_2}.
If the input matrix has $(r=3)$, both algorithms do not perform well since the Tanner graph is irregular.
	%
	For a particular error $e=(0 0 0 0 \omega^2 0 0)$, both decoders converge to a large-weight $\hat e = (0 0 \omega^2 \omega^2 \omega^2 \omega^2 0)$ due to overestimated messages.
	This is improved if the decoding is based on $(r=7)$. 
In \cite{Bab+15}, Babar {\it et~al.} considered a fixed initialization $\epsilon_0=0.26$ for \eqref{eq:init}. 
We adopt a similar value $\epsilon_0\approx 0.24$ and it indeed improves for the case of $(r=3)$ when $q=2$, as also shown in Fig.~\ref{fig:713_2}.

\begin{figure} 
	\centering \includegraphics[width=0.54\textwidth]{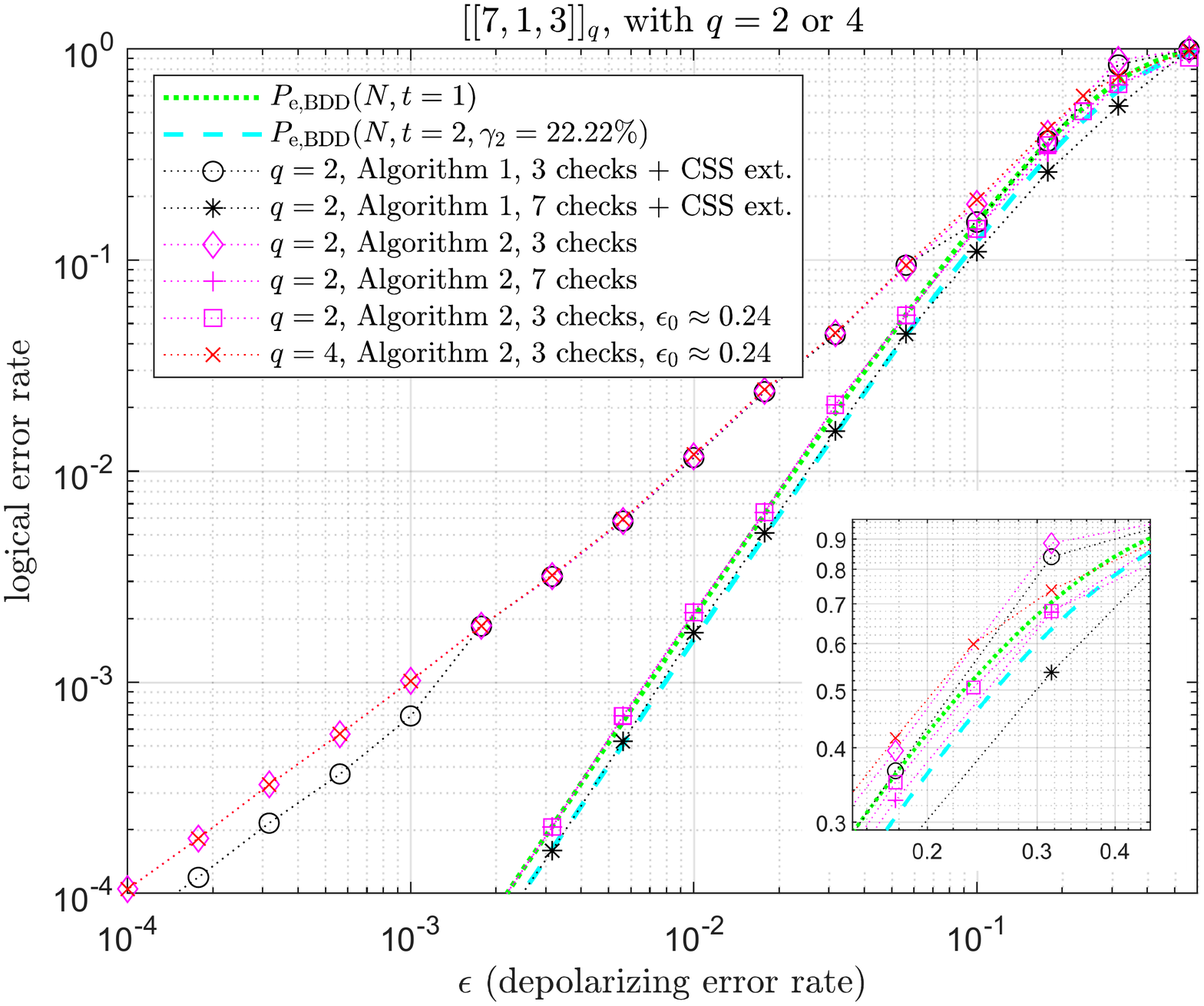}
	\caption{
		Decoding $[[7,1,3]]_q$ codes for $q=2$ and $q=4$ with a maximum number of 10 iterations, over the channel in Def.~\ref{def:ch}.
		When $q=4$, the curves are similar to those of $q=2$, except for one case as shown in the figure.
	} \label{fig:713_2}
\end{figure}

\begin{figure}
	\centering \includegraphics[width=0.48\textwidth]{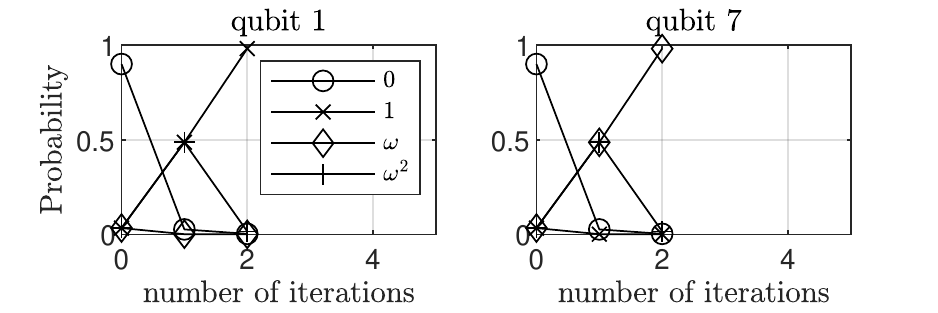}
	\caption{Decoding on error $(1 0 0 0 0 0 \omega)$ by Algorithm~\ref{alg:LLR-BP} with $q=2$: successful.
	} \label{fig:713_LlrBP}	\vspace*{\floatsep}
	\centering \includegraphics[width=0.48\textwidth]{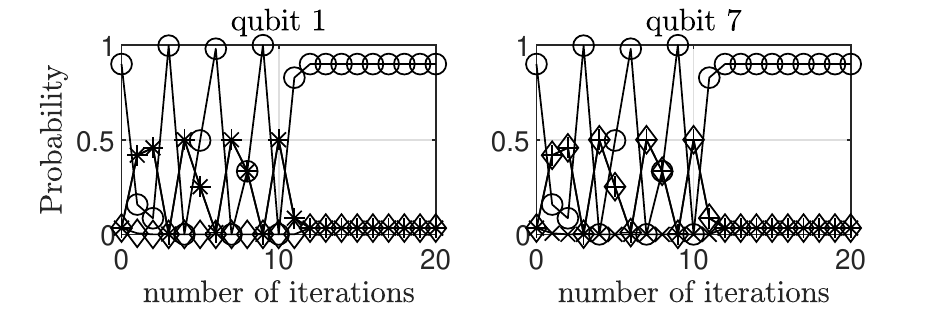}
	\caption{Decoding on error $(1 0 0 0 0 0 \omega)$ by Algorithm~\ref{alg:CBP} with $q=2$: failed.
	The convention of labels is the same as that in Fig.~\ref{fig:713_LlrBP}.
	The decoder runs to a large-weight state $(111111\omega)$ at iteration 10 after hard-decision, and will be trapped around $(0000000)$ after iteration 10.
	} \label{fig:713_ClsBP}	
\end{figure}

Note that, when ${q=2}$, the code has $2^{6}$ different error syndromes, corresponding to the zero vector (no error), \mbox{$21$ weight-one} errors, and \mbox{$42$ weight-two} errors. Thus $\gamma_2 = \frac{42}{{7\choose 2}\times 3^2} \approx 22.22\%$ in~\eqref{eq:BDD}.
Algorithm~\ref{alg:LLR-BP} is able to achieve this (optimum) correction capability, while Algorithm~\ref{alg:CBP} cannot.
This is indicated by two BDD curves in Fig.~\ref{fig:713_2}.
Note that for large $\epsilon$, Algorithm~\ref{alg:LLR-BP} with $(r=7)$ has performance better than $P_\text{e,BDD}(N,t=2,\gamma_2=22.22\%)$ 
because of degeneracy.\footnote{
	Given a specific syndrome, the decoder may always output the same low-weight error, but for large $\epsilon$, many high-weight degenerate errors occur with high probabilities and will be counted as decoding success.}

For a specific weight-two error $(1 0 0 0 0 0 \omega)$, we plot the decoding output probabilities based on $(r=7)$ and $\epsilon_0=0.1$ 
in Figs.~\ref{fig:713_LlrBP} and~\ref{fig:713_ClsBP} for Algorithms~\ref{alg:LLR-BP} and~\ref{alg:CBP}, respectively. 
Algorithm~\ref{alg:LLR-BP} successfully converges, while Algorithm~\ref{alg:CBP} is trapped around the zero vector.
(Algorithm~\ref{alg:CBP} runs into a large-weight error before trapped, due to overestimated messages.)

Next, consider the case of $q=4$. 
The (binary) $\tilde H^{(r)}$ still has $r$ rows, but the ($q^2$-ary) $H^{(r)}$ has $4r$ rows
and has much more short cycles compared to the case of $q=2$ ($2r$ rows).
However, the results are similar except that Algorithm~\ref{alg:CBP} with $(r=3)$ and $\epsilon_0 \approx 0.24$ does not perform well.
We plot this case in Fig.~\ref{fig:713_2}.
(We scan many values of $\epsilon_0$, including the value 0.26 used in \cite{Bab+15}, and also try $\epsilon_0=\epsilon$, 
but none of the values provide the improvement like the case of $q=2$.)




\ifCLASSOPTIONcaptionsoff		
\newpage						
\fi								



\EOD

\end{document}